\documentclass[journal]{IEEEtran}

\usepackage{graphicx,color}

\usepackage{amsmath,amssymb,amsfonts,bm}
\usepackage{algorithmic}
\usepackage{mdwmath}
\usepackage{booktabs}

\usepackage{amssymb}
\usepackage{amsmath}
\usepackage{amsthm}
\usepackage{mathrsfs}
\usepackage{mathtools}

\usepackage{textcomp}
\usepackage{float}

\usepackage{soul,color}
\usepackage{cancel}
\usepackage{cite}

\usepackage{caption}
\usepackage[subrefformat=parens]{subcaption}
\theoremstyle{plain}
\newtheorem{thm}{Theorem}

\newtheorem{defn}[thm]{Definition}

\newcommand{\Figref}[1]{Fig.~\ref{#1}}
\newcommand{\Tabref}[1]{Table~\ref{#1}}

\newcommand{\norm}[1]{\left\lVert#1\right\rVert}
\newcommand{\normf}[1]{\left\lVert#1\right\rVert_F}

\newcommand{\Log}[1]{\operatorname{Log} #1 }
\newcommand{\Exp}{\operatorname{Exp}}

\newcommand{\tr}{\operatorname{tr} }
\newcommand{\argmin}{\operatorname{arg}\min }
\newcommand{\trans}{^{\operatorname{T}}}
\newcommand{\hermconj}{^{\mathsf{H}}}
\newcommand{\frechetderi}{\left.\frac{d}{d\varepsilon}\right|_{\varepsilon = 0}}
\newcommand{\riegrad}{\operatorname{grad}}

\newcommand{\inpror}[1]{\left\langle#1\right\rangle_{\bm{R}}}
\newcommand{\inprop}[1]{\left\langle#1\right\rangle_{\bm{P}}}
\newcommand{\inproi}[1]{\left\langle#1\right\rangle_{\bm{I}}}

\newcommand{\Lyapunov}[2]{ \mathcal{L}_{#1}\left[ #2 \right] }
\newcommand{\vect}{\operatorname{vec}}

\ifCLASSINFOpdf

\else

\fi

\begin{document}
%
% paper title
% can use linebreaks \\ within to get better formatting as desired
\title{The Comparison of Riemannian Geometric Matrix-CFAR Signal Detectors}

\author{
        Yusuke~Ono,~\IEEEmembership{Member,~IEEE,} and 
        Linyu~Peng,~\IEEEmembership{Member,~IEEE}% <-this % stops a space
\thanks{This work was partially supported by Japan Science and Technology Agency SPRING (No. JPMJSP2123), Japan Society for the Promotion of Science KAKENHI (No. JP20K14365), Japan Science and Technology Agency CREST (No. JPMJCR1914), and the Fukuzawa Fund and KLL of Keio University.   {\it (Corresponding author: Linyu Peng.)}

Y. Ono and L. Peng are with the Department of Mechanical Engineering, Keio University, Yokohama 223-8522, Japan (e-mail: yuu555yuu@keio.jp; l.peng@mech.keio.ac.jp).
}}
%\date{January 19, 2022}
% The paper headers
%\markboth{IEEE Transactions on Communications}%
%{Submitted paper}

% make the title area
\maketitle

% \editor{This work is licensed under a Creative Commons Attribution 4.0 License. For more information, see https://creativecommons.org/licenses/by/4.0/}
% \supplementary{Color versions of one or more of the figures in this article are available online at \href{http://ieeexplore.ieee.org}{http://ieeexplore.ieee.org}.}
% \markboth{AUTHOR ET AL.}{SHORT ARTICLE TITLE}

\begin{abstract}
    Essential characteristics of signal data can be captured by the autocovariance matrix, which, in the stationary scenarios, is Toeplitz Hermitian positive definite (HPD).
    In this paper, several well-known Riemannian geometric structures of HPD matrix manifolds are applied to signal detection, including the affine invariant Riemannian metric, the log-Euclidean metric, and the Bures--Wasserstein (BW) metric, the last of which was recently extended to HPD manifolds. 
    Riemannian gradient descent algorithms are proposed to solve the corresponding geometric means and medians, that play fundamental roles in the detection process. 
    Simulations within the scenario using the ideal steering vector as the target signal provide compelling evidence that the BW detectors outperform the other geometric detectors as well as the conventional adaptive matched filter and adaptive normalized matched filter when observation data are limited. 
    Further simulations demonstrate that the matrix-CFAR is robust in scenarios where the signal is mismatched.
    In addition to detection performances, robustness of the geometric detectors to outliers and computational complexity of the algorithms are analysed. 
\end{abstract}


\section{INTRODUCTION}
In signal detection, the constant false alarm rate (CFAR) method has gained wide recognition as a conventional technique for suppressing clutter \cite{RadarSignalProcessing}.
The cell-averaged CFAR (CA-CFAR), known as a popular CFAR method, that employs the Fourier transform and utilises the Doppler spectral density for differentiating targets from clutter which arises from the backscattering of objects, however, tends to perform inadequately in the case of short waveforms with dense and non-homogeneous clutter \cite{weiss1982analysis}. 
%The clutter arises from the backscattering of objects which are not of tactical relevance to the target detection.
To improve discrimination between targets and clutter, one statistical method is to test the likelihood ratio by using an estimator of the autocovariance matrix. 
A classical estimator of the autocovariance matrix is the sample covariance matrix (SCM), which is the maximum likelihood estimator \cite{SCMbyMLE}. 
However, it can be accurately estimated only when a sufficient number of observation data are independent and identically distributed \cite{4101326}.
In order to achieve efficient detection performance within non-homogeneous clutter while reducing the required number of observation data, several methods utilising a priori information have been proposed in \cite{5417154, LI20191, s21072391, 6853408, ctx11663242640004034}. 
Additionally, in \cite{5484507, 4154721}, the autocovariance matrix was estimated utilising the Bayesian methods. 
The strength of clutter varies depending on the condition of surroundings, transmitted frequency, and other factors \cite{ModernRadar}. 
Therefore, in practical implementations, the clutter environment can be range-dependent and non-stationary. In such situations, it is often challenging to obtain secondary data that are free from target components.
Accurately capturing the statistical property about clutter precedently can be difficult, and observation data are often contaminated. 
Insufficient information about clutter environment can make detection performance degraded in the methods using a priori information.

A novel clutter suppression method, the matrix-CFAR, was recently proposed and developed, which does not require any a priori information and leverages the autocovariance matrix of the signals \cite{BarInoTool, ETIVC2009}. 
This method operates under the premise that the autocovariance matrix, which characterizes the autocorrelation of observation data, is Hermitian positive-definite (HPD) and 
has been leveraged for detection problems \cite{7560351, DroneDetect}. 
Its detection efficiency has been demonstrated in the observation of wake eddy turbulence \cite{BARBARESCO201054}, Burg estimate methods of radar scatter matrices \cite{7842633}, detection of X-band radar clutter \cite{BarHFandXband}, and so on. 
% \todo[inline]{suddenly AIRM appears: Riemannian geometry of HPD $->$ applied to detection $->$ however. 最初はAIRMで提案されたけど、それでは比較的にコストが高く他の計量も利用されている。}
In these applications, the affine invariant Riemannian metric (AIRM)  of the HPD manifolds was used. 
However, the computational cost of the AIRM is relatively high \cite{9817822}.
% because the computational complexity of the logarithm of $N \times N$ matrices, $\operatorname{Log}\bm{R} \sim O(N^4)$  
Recently, matrix-CFAR, based on various divergences has been proposed, which has shown improved detection performance and robustness to outliers in comparison with the AIRM \cite{MatrixCFARbyKL, CFARJensen,9764734}. 
% \todo[inline]{How about the comparison with LE?}
In \cite{fixpointHua}, the comparison between the log-Euclidean (LE) metric and the AIRM was conducted and their detection performances were shown almost the same. 
Regarding the computational complexity that will be studied in Section \ref{sec: CC}, while the LE mean can be derived by closed form, the AIRM mean has to be obtained numerically.
As structure of HPD manifolds is of great importance in the matrix-CFAR, various geometric structures of HPD manifolds have been proposed \cite{ctx15813500480004034, THANWERDAS2022101867, THANWERDAS2023163}. 

The Bures--Wasserstein (BW) distance, also known as the Earth mover's distance or Kantorovich--Rubinstein distance, was proposed in the context of optimal transport \cite{vaserstein1969markov,kantorovich1960mathematical}. It was then extended to function spaces and probability spaces \cite{fournier2015rate, 1011371118101}.
The BW distance has been applied in various fields, including quantum information \cite{luo2004informational, 2021NJPh}, optimisation theory \cite{han2021riemannian, villani2009optimal, ModinKlas2017} and machine learning \cite{peyre2019computational, pmlr-v70-arjovsky17a}.
In particular, in the field of machine learning, the BW distance is referred to as the Wasserstein distance and serves as a measure of dissimilarity in the space of mean-zero Gaussian densities \cite{takatsu2011wasserstein}.
In quantum information theory, it, referred to as the Bures distance, is utilised to quantify the dissimilarity between quantum states or density matrices \cite{spehner2013geometric}. In the current paper, the BW metric of HPD manifolds is applied to signal detection and compared with other Riemannian geometric structures. 
The contributions of this study are summarised as follows. 
\begin{itemize}
\item[(1)] The BW metric is applied to signal detection, where the BW mean and median are used to catch essential statistical properties of the observation data. The corresponding optimisation problems for BW mean and median are solved by Riemannian gradient descent algorithms; their computational complexity is studied and compared with other numerical algorithms.  
%and define the corresponding geometrical mean and median each of which is the solution to the optimisation problem. The optimisation problems in the Riemannian manifolds can be solved by the 
%Riemannian gradient descent algorithms are proposed to solve them numerically and then we study their computational complexity.
\item[(2)] Numerical detection performance of the BW detectors is compared with that of the AIRM, LE metric, adaptive matched filter (AMF) and adaptive normalized matched filter (ANMF) detectors as follows.
In one scenario, the ideal steering vector is regarded as the target signal, while in the other, we consider the presence of steering signal mismatches.
In the first scenario, the BW detectors outperform the other methods when observation data are limited.
To address the limitations of AMF, we propose AMFs by replacing the SCM with the Riemannian geometric means or medians.
 The relation between detection performance and the normalized Doppler frequency of target is further studied numerically.
In the second scenario, it is shown that the matrix-CFAR is more robust with respect to mismatched signals compared to the AMF and ANMF.
%Besides the theoretical analysis, through numerical simulations, their detection performance is investigated and compared with the means and medians associated with the AIRM, the log-Euclidean metric, the AMF as well as the ANMF. 
\item[(3)] Robustness of the detectors about outliers is analysed by the corresponding exact influence functions defined in \cite{onoTBDmedian}. 
\end{itemize}

This paper is structured as follows. In Section \ref{sec: mat}, the detection setting and theory are summarised. 
In Section \ref{sec: HPD}, after recalling the Riemannian geometry of HPD manifolds, such as the AIRM, the LE metric, and the BW metric, 
we introduce Riemannian geometric means and medians, as well as Riemannian gradient descent algorithms for obtaining them.
% for solving the corresponding optimisation problems. 
Besides, their computational complexity is investigated. 
Numerical detection performance of the BW detectors is conducted in Section \ref{sec: DP} compared with those using the AIRM and LE metric, and the AMF and ANMF. 
In Section \ref{sec: RA}, robustness analysis is conducted via the influence functions for all means and medians studied in this paper. 
% an orthogonal basis for Hermitian matrices is defined and applied to the computation of influence functions, allowing us to accurately illustrate the robustness of all means and medians.
We conclude and address future research lines finally in Section \ref{sec: con}.

\section{THE DETECTION SETTING AND THEORY}
\label{sec: mat}
In this section, we introduce the problem formulation of signal detection and then recall the conventional detection methods and the matrix-CFAR.
In this study, it is assumed that we detect the observation data from $N$ stationary channels. 
The detection problem will be modelled as a statistical hypothesis test: the null hypothesis is denoted as $H_0$ and represents the complex data $\bm{x}$ includes only clutter $\bm{c}$. On the other hand, the alternative hypothesis is denoted as $H_1$ and represents $\bm{x}$ includes both clutter $\bm{c}$ and the target $\bm{s}$, i.e.,
\begin{eqnarray}\label{eq: H01}
    \begin{cases}
        H_0:
        \begin{cases}
			\mathrm{target\ is\ absent,} \\
            \bm{x} = \bm{c}, \\
        \end{cases}  
            \vspace{0.2cm}\\
        H_1:
        \begin{cases}
			\mathrm{target\ is\ present,} \\
            \bm{x} = a  \bm{s} + \bm{c}. \\
        \end{cases}  \\
    \end{cases}
\end{eqnarray}
Here $a$ denotes the unknown amplitude coefficient. 
The observation data $\bm{x}$ is given by 
\begin{equation} \label{eq: received_signal}
    \bm{x} = \left(x_0, \ldots, x_{N-1}\right)^{\operatorname{T}}\in\mathbb{C}^N, \\
\end{equation}
where $\mathbb{C}^N$ is the $N$-dimensional complex space and $(\cdot)^{\operatorname{T}}$ represents the transpose of vectors or matrices.
The SCM of observation data $\left\{ \bm{x}_{i} \right\}_{i=1}^m$ is derived by \cite{SCMbyMLE} % the maximum likelihood estimation of the circularly symmetric complex Gaussian distribution.
\begin{equation}
    \label{eq: SCM}
	\bm{R}_{\mathrm{SCM}} = \frac1m \sum_{i=1}^m  \bm{x}_i \bm{x}_i\hermconj,\ \  \bm{x}_i \in \mathbb{C}^N .
\end{equation}
Here $(\cdot)\hermconj$ denotes the conjugate transpose of vectors or matrices. It has been commonly used; 
in signal processing, the AMF reads \cite{GLRT}
\begin{equation}
    \frac{\left| \bm{x}\hermconj \bm{R}^{-1}_{\mathrm{SCM}} \bm{s} \right|^2}{\bm{s}\hermconj\bm{R}^{-1}_{\mathrm{SCM}} \bm{s} } \overset{H_1}{\underset{H_0}{\gtrless}} \gamma 
\end{equation}
and the ANMF reads \cite{381910}
\begin{equation}
    \frac{\left| \bm{x}\hermconj \bm{R}^{-1}_{\mathrm{SCM}} \bm{s} \right|^2}{\left( \bm{x}\hermconj\bm{R}^{-1}_{\mathrm{SCM}} \bm{x} \right)\left( \bm{s}\hermconj\bm{R}^{-1}_{\mathrm{SCM}} \bm{s} \right)} \overset{H_1}{\underset{H_0}{\gtrless}} \gamma, 
\end{equation}
where $\gamma$ denotes the detection threshold. %In these methods, it is required to estimate the clutter autocovariance matrix

In the matrix-CFAR, to decide whether a target $\bm{s}$ is present within the cell under test (CUT) or not, the autocovariance matrix of the observation data can be modelled by a Toeplitz HPD matrix: \cite{6514112,40004034} 
\begin{equation}\label{eq: covariancematrix}
	\begin{aligned}
	\hspace{-7mm}\bm{R} = 
	    \begin{pmatrix}
		    r_0 & \cdots & \overline{r}_k & \cdots & \overline{r}_{N-1} \\
		    \vdots & \ddots & \ddots & \ddots & \vdots \\
		    r_k & \ddots & r_0 & \ddots & \overline{r}_k \\
		    \vdots & \ddots & \ddots & \ddots & \vdots \\
		    r_{N-1} & \cdots & r_k & \cdots & r_0 \\
	    \end{pmatrix},
	    % \in \mathscr{P}(N,\mathbb{C}) ,\\
	\end{aligned}
\end{equation}
where $ 0 \le k \le N-1$. 
% where $ \mathscr{P}(N,\mathbb{C}) $ represents the space of $N$-dimensional HPD matrices.
Its component is the $k$-th autocorrelation function given by
\begin{equation}
		r_k = \mathbb{E} \left[x_l \overline{x}_{l+k} \right], \ \ 0 \le l \le N-k-1.
\end{equation}
Here $\mathbb{E}[\cdot]$ is the statistical expectation, and $\overline{r}_{k}$ is the conjugate of $r_{k}$. 
It is known that the Toeplitz structure leads to the improvement in an estimate of the autocovariance matrix \cite{ABY2012,9054969,987665}.
The $k$-th autocorrelation function $r_{k}$ can be estimated by the observation data as
\begin{eqnarray}
	\label{eq: autoregressive}
		r_k = \dfrac{1}{N} \sum\limits_{l = 0} ^{N-k-1} x_l \overline{x}_{l+k} ,\ \ 0 \le k \le N-1. 
\end{eqnarray}
 Equation \eqref{eq: autoregressive} to estimate autocorrelation coefficients is not robust, and under assumption of local stationarity of signal, these autocorrelation coefficients could be estimated by regularised Burg algorithm \cite{7083263}.  
An estimator of the clutter autocovariance matrix is derived as $\bm{R}_{g}$ by autocovariance matrices of the observation data $\{ \bm{R}_{i}\}_{i=1}^m$. 
As a consequence, the detection problem \eqref{eq: H01} can be formulated as 
\begin{eqnarray}
    \begin{cases}
        H_0:\bm{R}_{\mathrm{CUT}} = \bm{R}_{g}, \vspace{0.2cm} \\
        H_1:\bm{R}_{\mathrm{CUT}} \neq \bm{R}_{g}. \\
    \end{cases}
\end{eqnarray}
Here the matrix $\bm{R}_{\mathrm{CUT}}$ is the autocovariance matrix within the CUT. 
Considering a threshold $\gamma$, we can approach the detection problem as the task of distinguishing $\bm{R}_{\mathrm{CUT}}$ from $\bm{R}_{g}$ by 
\begin{eqnarray}\label{eq:decison_of_R_g}
    d\left( \bm{R}_g, \bm{R}_{\mathrm{CUT}}\right) \overset{H_1}{\underset{H_0}{\gtrless}} \gamma
\end{eqnarray}
where $d\left( \bm{R}_{g}, \bm{R}_{\mathrm{CUT}}\right)$ is either a geodesic distance or a divergence. In the current study, we focus on geodesic distances only, which will be introduced in the following section. 
% difference between $\bm{R}_{\mathrm{CUT}}$ and $\bm{R}_{g}$, such as
\Figref{fig: processMCFAR} shows the process of the matrix-CFAR which is described as follows. Firstly we compute autocovariance matrices $\{ \bm{R}_i\}_{i=1}^m$ of the observation data $\{ \bm{x}_i\}_{i=1}^m$ using \eqref{eq: covariancematrix} and \eqref{eq: autoregressive}. 
Riemannian geometric mean or median $\bm{R}_g$ of $\{ \bm{R}_{i}\}_{i=1}^m$ can be derived as the solutions of optimisation problems about the geodesic distances. 
The matrix $\bm{R}_g$ can be regarded as an estimator of the clutter autocovariance matrix since almost all observation data $\{ \bm{x}_{i}\}_{i=1}^m$ are considered as clutter, that is of the essence. 
The detection decision is completed by comparing the distance of $\bm{R}_{\mathrm{CUT}}$ and $\bm{R}_g$ with the threshold $\gamma$. 

\begin{figure}[htbp]
	\centering
	\includegraphics[width = \linewidth]{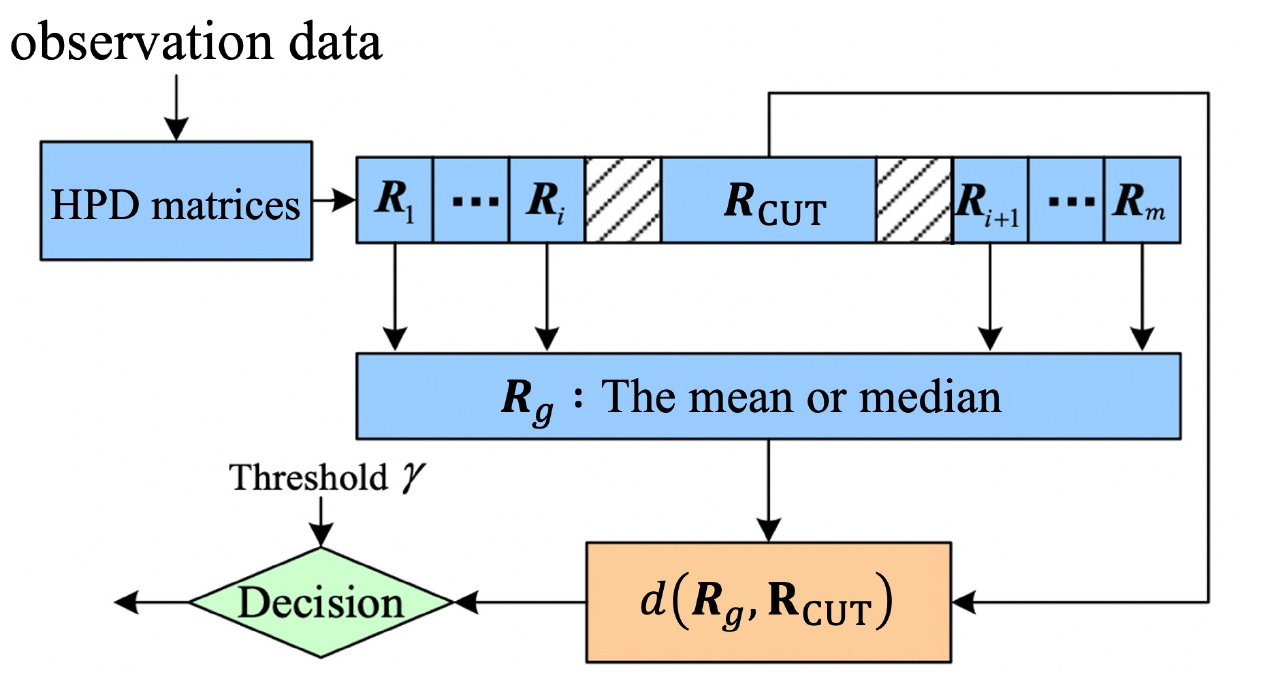}
	\caption{The process of matrix-CFAR.}\label{fig: processMCFAR}
\end{figure}

\section{THE GEOMETRY OF HPD MATRIX MANIFOLDS}
\label{sec: HPD}
In this section, we review Riemannian geometry of HPD manifolds equipped with the AIRM, the LE metric, and the BW metric. 
The corresponding Riemannian geometric means and medians are defined and solved numerically by Riemannian gradient descent algorithms.

\subsection{The AIRM}
\label{sec: AIRM}
% We define the AIRM on the HPD matrix manifold $\mathscr{P}(N,\mathbb{C})$, and introduce the corresponding geodesic distance. 
The $N$-dimensional HPD matrix manifolds are represented by $ \mathscr{P}(N,\mathbb{C}) $ and
the tangent space at a point  $\bm{P}\in \mathscr{P}(N,\mathbb{C})$ is the space of $N$-dimensional Hermitian matrices denoted by $T_{\bm{P}}\mathscr{P}(N,\mathbb{C}) \cong \mathscr{H}(N,\mathbb{C})$. 
The tangent bundle is represented by $T\mathscr{P}(N,\mathbb{C})=\cup_{\bm{P}}T_{\bm{P}}\mathscr{P}(N,\mathbb{C})$.
The AIRM is defined by
\begin{equation}\label{eq: AIRM}
\begin{aligned}
\langle \bm{A},\bm{B} \rangle_{\bm{P}} := \tr \left(\bm{P}^{-1}\bm{A}\bm{P}^{-1}\bm{B}\right) % \langle \bm{P}^{-1/2}\bm{A}\bm{P}^{-1/2}, \bm{P}^{-1/2}\bm{B}\bm{P}^{-1/2} \rangle.
\end{aligned}
\end{equation}
for  $ \bm{A},\bm{B}\in T_{\bm{P}}\mathscr{P}(N,\mathbb{C})$.
% This inner product leads to the AIRM distance on the HPD matrix manifold $\mathscr{P}(N,\mathbb{C})$ that is given at each $\bm{P}$ by the infinitesimal arclength
% \begin{equation}
% ds^2 = \tr \left( \bm{P}^{-1}d\bm{P}\bm{P}^{-1}d\bm{P} \right) .
% \end{equation}
The geodesic between two points $\bm{P}_1, \bm{P}_2 \in \mathscr{P}(N,\mathbb{C}) $ is given by
\begin{equation}\label{eq:geocurve}
    \begin{aligned}
        \bm{P}(t) % & = \bm{P}_0^{\frac12} \exp \left( t \Log \left( \bm{P}_0^{-\frac12} \bm{P}_1 \bm{P}_0^{-\frac12} \right) \right) \bm{P}_0^{\frac12} \\ 
        & = \bm{P}_1^{\frac12} \left( \bm{P}_1^{-\frac12} \bm{P}_2 \bm{P}_1^{-\frac12} \right)^t \bm{P}_1^{\frac12},\ t\in[0,1],       
    \end{aligned}
\end{equation}
% Obviously, we have $\bm{P}(0)  = \bm{P}_0 ,$  $\bm{P}(1) = \bm{P}_1 $,
and its tangent vector at the point $\bm{P}_1$ is 
\begin{equation}
    \dot{\bm{P}} (0) = \bm{P}_1^{\frac12} \Log \left( \bm{P}_1^{-\frac12} \bm{P}_2 \bm{P}_1^{-\frac12} \right) \bm{P}_1^{\frac12}. \\ 
\end{equation}
% \begin{equation}
%     \dot{\bm{P}} (t) = \bm{P}_0^{\frac12} \exp \left( t \Log \left( \bm{P}_0^{-\frac12} \bm{P}_1 \bm{P}_0^{-\frac12} \right) \right) \Log \left( \bm{P}_0^{-\frac12} \bm{P}_1 \bm{P}_0^{-\frac12} \right) \bm{P}_0^{\frac12}. \\ 
% \end{equation}
% and
% \begin{equation}
%     \bm{P}^{-1}(t) \dot{\bm{P}} (t) = \bm{P}_0^{-\frac12} \Log \left( \bm{P}_0^{-\frac12} \bm{P}_1 \bm{P}_0^{-\frac12} \right) \bm{P}_0^{\frac12}. \\ 
% \end{equation}
Under the AIRM, the geodesic distance connecting two points $\bm{P}_1, \bm{P}_2 \in \mathscr{P}(N,\mathbb{C}) $ is \cite{RiegeoMoa}
\begin{equation}
    \begin{aligned}\label{eq:Riedis}
		d_{\text{AIRM}}\left( \bm{P}_1, \bm{P}_2 \right) = \normf{ \Log \left( \bm{P}_1^{-\frac12} \bm{P}_2 \bm{P}_1^{-\frac12} \right) },
		% & = \int ds \\ 
		% & = \int_0^1 \left( \tr \left( \bm{P}(t)^{-1} \dot{\bm{P}}(t) \bm{P}(t)^{-1} \dot{\bm{P}}(t) \right) \right)^{\frac12} dt \\ 
		% & = \int_0^1 \left( \tr \left( \bm{P}_0^{-\frac12} \Log \left( \bm{P}_0^{-\frac12} \bm{P}_1 \bm{P}_0^{-\frac12} \right) \bm{P}_0^{\frac12}\right)^2 \right)^{\frac12} dt \\
		% & = \left( \tr \left( \bm{P}_0^{-\frac12} \Log \left( \bm{P}_0^{-\frac12} \bm{P}_1 \bm{P}_0^{-\frac12} \right) \bm{P}_0^{\frac12}\right)^2 \right)^{\frac12} \\
        % & = \left( \tr \left( \Log^2 \left( \bm{P}_0^{-\frac12} \bm{P}_1 \bm{P}_0^{-\frac12} \right) \right) \right)^{\frac12} \\
		% & = \norm{ \Log \left( \bm{P}_0^{-\frac12} \bm{P}_1 \bm{P}_0^{-\frac12} \right) } = \norm{ \Log \left(  \bm{P}_0^{-1} \bm{P}_1 \right) }, 
    \end{aligned}
\end{equation}
where $\normf{\bm{A}}:=\sqrt{\tr \left( \bm{A} \bm{A}\hermconj \right) }$ is the Frobenius norm. % associated to the Frobenius metric $\left\langle \bm{A}, \bm{B} \right\rangle := \tr{\left( \bm{A}\hermconj \bm{B}\right)}$. 
The exponential map in the Riemannian manifold $\mathscr{P}(N,\mathbb{C})$ denotes $\Exp: T\mathscr{P}(N,\mathbb{C})  \rightarrow \mathscr{P}(N,\mathbb{C})$ defined by using geodesics. 
In case of the AIRM, considering $\bm{P}(t)$ as the geodesic \eqref{eq:geocurve} between two arbitrary HPD matrices $\bm{P}_1=\bm{P}(0),\bm{P}_2=\bm{P}(1)$, the exponential map leads to the endpoint $\bm{P}_2$ of $\bm{P}(t)$, namely 
\begin{equation}\label{eq:ExpAIRM}
    \begin{aligned}
        \Exp \left( \bm{P}_1, \dot{\bm{P}}(0) \right) % & = \bm{P}_1 \\
		% & = \bm{P}_0^{\frac12} \exp \left( \Log \left( \bm{P}_0^{-\frac12} \bm{P}_1 \bm{P}_0^{-\frac12} \right) \right) \bm{P}_0^{\frac12} \\
        & = \bm{P}_1^{\frac12} \exp \left( \bm{P}_1^{-\frac12} \dot{\bm{P}}(0) \bm{P}_1^{-\frac12}\right) \bm{P}_1^{\frac12}.\\
    \end{aligned}
\end{equation}
At the point $\bm{P}_1$, the logarithm map $\Log : \mathscr{P}(N,\mathbb{C})\times \mathscr{P}(N,\mathbb{C}) \rightarrow \mathscr{H}(N,\mathbb{C})$ is the inverse of the exponential map, i.e.,
\begin{equation}\label{eq: LogAIRM}
    \begin{aligned}
        \Log \left( \bm{P}_1, \bm{P}_2 \right) % & = \dot{\bm{P}}(0) \\
        % \Logm{\bm{P}_0}{\bm{P}_1} & =\dot{\bm{P}}(0) \\
		& = \bm{P}_1^{\frac12} \Log \left( \bm{P}_1^{-\frac12} \bm{P}_2 \bm{P}_1^{-\frac12} \right) \bm{P}_1^{\frac12}. 
    \end{aligned}
\end{equation}
% It is also denoted as $\Log(\bm{P}_0, \bm{P}_1)$.
The midpoint/mean between $\bm{P}_1$ and $\bm{P}_2$ is denoted by $\bm{P}_1 \# \bm{P}_2 $ and given by the Pusz--Woronowicz formula \cite{PUSZ1975159}
\begin{equation}\label{eq:Riemid}
    \bm{P}_1 \# \bm{P}_2 = \bm{P}\left(t = \frac{1}{2}\right) = \bm{P}_1^{\frac12} \left( \bm{P}_1^{-\frac12} \bm{P}_2 \bm{P}_1^{-\frac12} \right)^{\frac{1}{2}} \bm{P}_1^{\frac12}.
\end{equation}
This mean is symmetric with respect to $\bm{P}_1$ and $\bm{P}_2$ and is the only solution of the Riccati equation \cite{BHATIA2019165}
\begin{equation}
    \bm{X} \bm{P}_1^{-1} \bm{X} = \bm{P}_2.
\end{equation}

\subsection{The LE Metric}
\label{sec: LEM}
The LE metric has been applied to, for instance, the tensor filed and image data \cite{arsigny2006log, AFPA2007, huang2015log}.
In radar detection, it has been utilised for the outlier rejection \cite{6573681, 6825699, 10057028, 10168102}. 
Some important structures and properties of  HPD manifolds under the LE metric are summarised in \Figref{fig: LEM}.

\begin{figure*}[htbp]
    \centering
    \includegraphics[width = \linewidth]{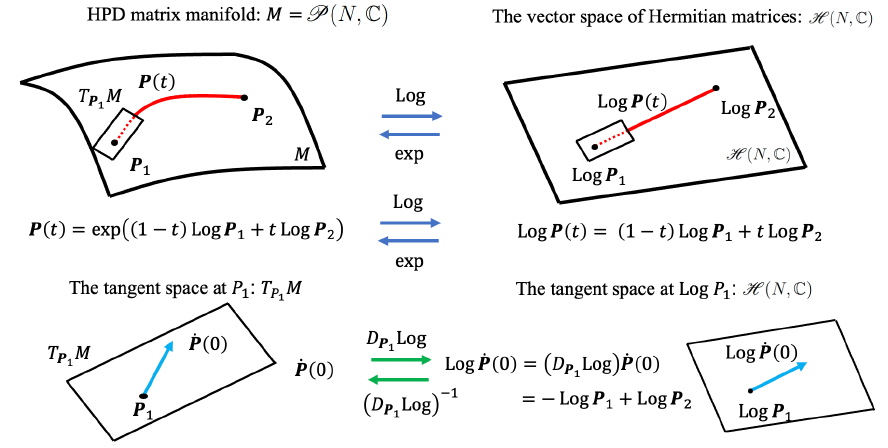}
    \caption{The HPD manifold under the LE metric.}
    \label{fig: LEM}
\end{figure*}

The LE metric is derived by leveraging the Lie group structure of $\mathscr{P}(N,\mathbb{C})$ under the group operation 
\begin{equation}
    \bm{P}_1 \odot \bm{P}_2 := \exp \left( \Log \bm{P}_1 + \Log \bm{P}_2  \right)
\end{equation}
for $\forall \bm{P}_1, \bm{P}_2 \in \mathscr{P}(N,\mathbb{C})$, 
where $\exp$ and $\operatorname{Log}$ denote the matrix exponential and logarithm; the latter is unique for HPD matrices, and should not be confused with the logarithmic map \eqref{eq: LogAIRM}.
% The LEM on the Lie group of SPD matrices corresponds to an Euclidean metric in the SPD matrix logarithmic domain. 
The LE metric of two elements $\bm{V}_1, \bm{V}_2$ in the tangent space at a point $\bm{P}$ is defined as
\begin{equation}\label{eq:LEM}
    \inprop{\bm{V}_1, \bm{V}_2} = \inproi{ \left( D_{\bm{P}} \Log\right) \bm{V}_1, \left( D_{\bm{P}} \Log\right) \bm{V}_2},
\end{equation}
where $ \left( D_{\bm{P}} \Log\right) \bm{V}_1$ is the directional derivative of the matrix logarithm at $\bm{P}$ along $\bm{V}$ and $\bm{I}$ denotes the $N$-dimensional identity matrix.
The geodesic between $ \bm{P}_1$ and $ \bm{P}_2 $ is given by 
\begin{equation}
    \bm{P}(t) = \exp \left( (1-t) \Log \bm{P}_1 + t \Log \bm{P}_2 \right).
\end{equation}
The exponential and logarithm maps for the LE metric are respectively given by  
\begin{eqnarray}
    \Exp \left( \bm{P}_1, \dot{\bm{P}}(0) \right) = \exp \left( \Log \bm{P}_1 + \left( D_{\bm{P}} \Log\right) \dot{\bm{P}}(0) \right), \label{eq: ExpLEM} \\
    \Log \left( \bm{P}_1, \bm{P}_2 \right) = \left( D_{\Log \bm{P}_1} \exp\right) \left( \Log \bm{P}_2 - \Log \bm{P}_1 \right), \label{eq: LogLEM}
\end{eqnarray}
where
\begin{equation}
\dot{\bm{P}}(0)=\exp\left(\Log \bm{P}_2-\Log \bm{P}_1\right).
\end{equation}
Using \eqref{eq:LEM}, \eqref{eq: ExpLEM} and \eqref{eq: LogLEM}, the geodesic distance between $\bm{P}_1$ and $ \bm{P}_2$ is derived as 
\begin{equation*}\label{eq: LEdis}
    \begin{aligned}
        d^2_{\text{LE}}(\bm{P}_1,\bm{P}_2) & = \left\langle \left( D_{\Log \bm{P}_1} \exp\right) \left( \Log \bm{P}_2 - \Log \bm{P}_1 \right), \right.\\ 
        & ~~~~~~~ \left. \left( D_{\Log \bm{P}_1} \exp\right) \left( \Log \bm{P}_2 - \Log \bm{P}_1 \right) \right\rangle_{\bm{P}_1} \\ 
        & = \inproi{\Log \bm{P}_2 - \Log \bm{P}_1 , \Log \bm{P}_2 - \Log \bm{P}_1 } \\
        & = \normf{\Log\bm{P}_2 - \Log\bm{P}_1}^2. \\
    \end{aligned}
\end{equation*}

\subsection{The BW Metric}
\label{sec: BWM}
Here we introduce the BW metric on $\mathscr{P}(N,\mathbb{C})$, and the BW distance between HPD matrices. 
% The BW distance has applications in various fields, including quantum information \cite{luo2004informational,spehner2013geometric}, optimisation theory \cite{han2021riemannian} and optimal transport \cite{peyre2019computational}.
% In particular, in the domain of optimal transport, this distance is referred to as the Wasserstein distance and serves as a measure of dissimilarity in the space of mean-zero Gaussian densities.
% In quantum information theory, this distance, referred to as the Bures distance, is utilized to quantify the dissimilarity between quantum states or density matrices.
The BW metric at a point $\bm{P}$ is defined by \cite{BHATIA2019165,BW_vOJ}
\begin{equation} \label{eq: BWmetric}
    \langle \bm{A},\bm{B} \rangle_{\bm{P}} := \frac{1}{2} \tr \left(\Lyapunov{\bm{P}}{\bm{A}} \bm{B}\right) % \langle \bm{P}^{-1/2}\bm{A}\bm{P}^{-1/2}, \bm{P}^{-1/2}\bm{B}\bm{P}^{-1/2} \rangle.
\end{equation}
% \begin{equation}
%   \langle \bm{A},\bm{B} \rangle_{\bm{P}} := \frac{1}{2} \operatorname{Re} \tr \left(\Lyapunov{\bm{P}}{\bm{A}} \bm{B}\right) % \langle \bm{P}^{-1/2}\bm{A}\bm{P}^{-1/2}, \bm{P}^{-1/2}\bm{B}\bm{P}^{-1/2} \rangle.
% \end{equation}
for  $ \bm{A},\bm{B}\in T_{\bm{P}}\mathscr{P}(N,\mathbb{C})$, where $\mathcal{L}_{\bm{P}}$ is called the Lyapunov operator and $\mathcal{L}_{\bm{P}}[\bm{A}]$ is the solution of the Lyapunov equation 
\begin{equation}
    \bm{P} \Lyapunov{\bm{P}}{\bm{A}} + \Lyapunov{\bm{P}}{\bm{A}}\bm{P} = \bm{A} .    
\end{equation}
The geodesic between two points $\bm{P}_1, \bm{P}_2 \in \mathscr{P}(N,\mathbb{C}) $ is derived by \cite{BHATIA2019165,BW_vOJ}
\begin{equation*}\label{eq: geocurveBW}
    \begin{aligned}
        \bm{P}(t) % & = \left( 1 - t \right)^2 \bm{P}_0 + t^2 \bm{P}_1 + t\left(1-t\right) \left(\left( \bm{P}_0 \bm{P}_1 \right)^{\frac12} + \left( \bm{P}_1 \bm{P}_0 \right)^{\frac12}  \right) \\
        & =  \left( 1 - t \right)^2 \bm{P}_1 + t^2 \bm{P}_2 \\
        & ~~~~~ + t\left(1-t\right) \left(\bm{P}_1 \left( \bm{P}_1^{-1} \# \bm{P}_2 \right) + \left( \bm{P}_1^{-1} \# \bm{P}_2 \right) \bm{P}_1 \right), \\
    \end{aligned}
\end{equation*}
% It is obvious that $\bm{P}(0)  = \bm{P}_0$, $\bm{P}(1) = \bm{P}_1$, 
% \begin{equation}
%     \dot{\bm{P}}(t) = -2(1-t)\bm{P}_0 + 2t \bm{P}_1 + (1-2t) \left(\bm{P}_0 \left( \bm{P}_0^{-1} \# \bm{P}_1 \right) + \left( \bm{P}_0^{-1} \# \bm{P}_1 \right) \bm{P}_0 \right),
% \end{equation}
and we have
\begin{equation}
    \dot{\bm{P}}(0) = \bm{P}_1 \left( \bm{P}_1^{-1} \# \bm{P}_2 - \bm{I} \right) + \left( \bm{P}_1^{-1} \# \bm{P}_2 - \bm{I} \right) \bm{P}_1 .
\end{equation}
The exponential map is
\begin{equation}
    \begin{aligned}
        & \Exp \left(\bm{P}_1, \dot{\bm{P}}(0)\right)  \\% & = \bm{P}_1 \\ 
        % & = \left(\bm{P}_0^{-1} \# \bm{P}_1 \right)  \bm{P}_0 \left(\bm{P}_0^{-1} \# \bm{P}_1 \right) \\
        % & = \left(\bm{I} +  \bm{P}_0^{-1} \# \bm{P}_1 - \bm{I}\right)  \bm{P}_0 \left(\bm{I} + \bm{P}_0^{-1} \# \bm{P}_1 - \bm{I} \right) \\
        & ~~~ = \left(\bm{I} + \Lyapunov{\bm{P}_1}{\dot{\bm{P}}(0)} \right)\bm{P}_1 \left(\bm{I} + \Lyapunov{\bm{P}_1}{\dot{\bm{P}}(0)} \right), \label{eq: ExpBW}
    \end{aligned}
 \end{equation}
while the logarithm map is
\begin{equation*}\label{eq: LogBW}
        \Log \left(\bm{P}_1, \bm{P}_2\right) % & = \dot{\bm{P}}(0) \\
        = \bm{P}_1 \left( \bm{P}_1^{-1} \# \bm{P}_2 - \bm{I} \right) + \left( \bm{P}_1^{-1} \# \bm{P}_2 - \bm{I} \right) \bm{P}_1 .
\end{equation*}
For the BW metric, the geodesic distance, called the BW distance, connecting two arbitrary points $\bm{P}_1, \bm{P}_2 $ is
%For more details, the reader may refer to \cite{RiegeoMoa}.
%The Wasserstein distance of two HPD matrices $\bm{P}_0, \bm{P}_1 $ was introduced \cite{} and defined by 
\begin{equation*}\label{eq: BWdis}
    d^2_{\text{BW}}\left(\bm{P}_1,\bm{P}_2\right) = \tr \left(\bm{P}_1\right) + \tr \left(\bm{P}_2\right) - 2\tr \left(\left( \bm{P}_1^{\frac12} \bm{P}_2 \bm{P}_1^{\frac12}\right)^{\frac12}\right).
\end{equation*}
 The BW geodesic distance connecting two arbitrary points has been introduced by Fr\'echet in \cite{frechet1957} (see also \cite{dowson1982frechet}).

\subsection{Riemannian Geometric Means and Medians}
Based on Riemannian geometry introduced in Subsections \ref{sec: AIRM}, \ref{sec: LEM} and \ref{sec: BWM}, the corresponding geometric means and medians of HPD matrices can be defined. 
\begin{defn}
    The Riemannian geometric mean of  the HPD matrices $\{ \bm{R}_i\}_{i=1}^m$ is defined by
    \begin{equation}\label{eq:geoMeanDef}
		\underset{ \bm{R} \in \mathscr{P}(N,\mathbb{C}) } \argmin \dfrac{1}{m}\sum\limits_{i=1}^{m} d^2(\bm{R}_i, \bm{R}),
    \end{equation}
    where the function $d$ is a geodesic distance on $\mathscr{P}(N,\mathbb{C})$. 
    The geometric median of the HPD matrices $\{ \bm{R}_i\}_{i=1}^m$ is defined by
    \begin{equation}\label{eq:geoMedianDef}
        \underset{ \bm{R} \in \mathscr{P}(N,\mathbb{C} )} \argmin \frac{1}{m} \sum_{i=1}^{m} d\left(\bm{R}_i,\bm{R} \right).
    \end{equation}
\end{defn}

For example, when the function $d$ is the Euclidean distance, the Riemannian (actually Euclidean) geometric mean is given by  
% is exactly the arithmetic mean of the HPD matrices $\{ \bm{R}_{i}\}_{i=1}^m$ given by 
\begin{equation}
    \underset{ \bm{R} \in \mathscr{P}(N,\mathbb{C} )} \argmin \dfrac{1}{m}\sum\limits_{i=1}^{m} \normf{\bm{R}-\bm{R}_i}^2,
\end{equation}
which yields
\begin{equation}\label{eq:arith}
    \dfrac{1}{m}\sum\limits_{i=1}^{m} \bm{R}_i.
\end{equation}
This is exactly the arithmetic mean.

%{\color{red} The geometric median of the HPD matrices is related to the geometric median for a probability measure on a Riemannian manifold which has been defined and whose uniqueness has been ensured through a specified condition in . }

 Generally, it can be quite challenging to show the uniqueness and existence of  geometric means and medians.  In \cite{yang2010riemannian}, the geometric median for a probability measure on a Riemannian manifold
was defined, with a natural condition to ensure its uniqueness.
% The convex characteristics of the objective function are known to have a connection with the curvature of the space.
The uniqueness and existence of the Riemannian geometric means under the AIRM and the LE metric were proved analytically (see, for instance, \cite{Moa2005, AFPA2007}).
% However, for the BW metric, the objective function of the Riemannian geometric mean is not convex. 
The BW mean is unique and is the solution of equation \eqref{eq: BWmean_solution} below \cite{BHATIA2019165, chewi2020gradient}. 
% 
% It is known that the HPD matrix manifold under the AIRM and BW metric is non-negatively curved \cite{luo2021geometric}.
% Result details are known for means (typically, barycenters) in non-negatively curved space, including the uniqueness, existence and so on; \cite{sturm2003probability}
The uniqueness and existence of the Riemannian geometric medians are related to the curvature of the space. 
Non-positive curvature leads to the unique median \cite{fletcher2009geometric}.
For the LE metric, the curvature is zero \cite{AFPA2007}, and hence, the LE median exists uniquely.
However, the curvatures with respect to the AIRM and the BW metric do not satisfy the non-positiveness property, that makes an analytic proof difficult.  
 Under the BW metric, specifically, the curvature is nonnegative and then uniqueness and existence of barycenter cannot be proved mathematically. 
In consequence, in some cases, barycenter computation could provide aberrant values.

In many cases, Riemannian geometrical means and medians cannot be obtained analytically and one often resorts to numerical methods such as the Riemannian gradient descent algorithm and the fixed-point algorithm \cite{6556702, Robuststatistics, Smith1993, Udr1994}.
% Riemannian gradient is the steepest direction in Riemannian manifolds.  % i.e., gradient descent algorithms on Riemannian manifolds
\vspace{0.2cm}

{\bf AIRM mean and median.} 
\begin{thm}\label{prop: AIRMm}
    The AIRM mean and median of the HPD matrices $\{ \bm{R}_{i}\}_{i=1}^m$ can respectively be derived by the Riemannian gradient descent algorithms as follows \cite{6556702, Robuststatistics}
    \begin{equation}\label{eq: gradAIRMmean}
        % \begin{aligned}
            \bm{R}_{t+1} 
            % & = \Exp \left( \bm{R}_{t}, - \eta_t \riegrad G \left( \bm{R}_{t} \right) \right) \\
            % & = \bm{R}_{t}^{\frac12} \exp \left( - \eta_t \bm{R}_{t}^{-\frac12} \riegrad G \left( \bm{R}_{t} \right) \bm{R}_{t}^{-\frac12} \right) \bm{R}_{t}^{\frac12} \\
            = \bm{R}_{t}^{\frac12} \exp \left( - \frac{2 \eta_t}{m} \sum_{i=1}^{m} \Log \left( \bm{R}_{t}^{\frac12} \bm{R}_{i}^{-1} \bm{R}_{t}^{\frac12}\right) \right) \bm{R}_{t}^{\frac12},
        % \end{aligned}
    \end{equation}
    and
%The AIRM median of the HPD matrices $\{ \bm{R}_{i}\}_{i=1}^m$ can be derived by the following Riemannian gradient descent algorithm,
    \begin{equation}\label{eq: gradAIRMmedian}
        \begin{aligned}
            \bm{R}_{t+1} 
            % & = \Exp \left( \bm{R}_{t}, - \eta_t \riegrad G \left( \bm{R}_{t} \right) \right) \\
            % & = \bm{R}_{t}^{\frac12} \exp \left( - \eta_t \bm{R}_{t}^{-\frac12} \riegrad G \left( \bm{R}_{t} \right) \bm{R}_{t}^{-\frac12}\right) \bm{R}_{t}^{\frac12} \\
            & = \bm{R}_{t}^{\frac12} \exp \left(- \frac{\eta_t}{m} \sum_{i=1}^{m} \frac{ \Log \left( \bm{R}_{t}^{\frac12} \bm{R}_{i}^{-1} \bm{R}_{t}^{\frac12}\right) }{ d_{\text{AIRM}}\left(\bm{R}_t ,\bm{R}_i \right) } \right) \bm{R}_{t}^{\frac12},
        \end{aligned}
    \end{equation}
    where  $\eta_t$ is the step size.
\end{thm}

{\bf LE mean and median.} 
\begin{thm}\label{th: LEmean}
    For the HPD matrices $\{ \bm{R}_i\}_{i=1}^m$, the LE mean can be derived in the closed form
    \begin{equation}
        \exp \left( \frac{1}{m} \sum^m_{i=1} \Log\bm{R}_i \right),
    \end{equation}
    a proof of which can be found in \cite{AFPA2007}.
\end{thm}
The exponential map of the LE metric \eqref{eq: ExpLEM} is not explicit, and instead, we utilise the fixed-point algorithm for the LE median \cite{MoaAvg}.

\begin{thm}\label{prop: LEmd}
    For $m$ HPD matrices $\{ \bm{R}_{i}\}_{i=1}^m$, the LE median can be derived by the fixed-point algorithm as follows:
    \begin{equation}\label{eq: LEmedianfixpoint}
        \begin{aligned}
            \bm{R}_{t+1} = \exp & \left( \left( \sum\limits^m_{i=1} \frac{ \Log\bm{R}_i}{d_{\text{LE}} \left(\bm{R}_i, \bm{R}_{t} \right)}  \right) \right. \\
            & ~~~~ \left. \times \left( \sum\limits^m_{i=1} \frac{1}{d_{\text{LE}} \left(\bm{R}_i, \bm{R}_{t} \right)} \right)^{-1} \right). \\
        \end{aligned}
    \end{equation}
\end{thm}
\begin{proof}
    The LE median is the solution of the optimisation problem of which the objective function $G(\bm{R})$ is given by 
    \begin{equation}
        G(\bm{R}) = \frac{1}{m} \sum_{i=1}^{m} d_{\text{LE}}\left(\bm{R}_i,\bm{R} \right).
    \end{equation}
    By using the Fr\'echet derivative \cite{LangSerge1995}, its Riemannian gradient, $\riegrad G \left( \bm{R} \right)$, is derived \cite{LogMat}, 
    \begin{equation*}
        \begin{aligned}
            & \inpror{\riegrad G \left( \bm{R} \right) , \bm{Y}} := \frechetderi G\left(\bm{R} + \varepsilon \bm{Y}\right) \\
            & = \frac{1}{m} \tr \left( \sum_{i=1}^m \int_0^1 \left( (\bm{R} - \bm{I})s+\bm{I}\right)^{-1} \right. \\
            & ~~~~~~~~~~~~ \left. \times \frac{\Log \bm{R} - \Log\bm{R}_i}{d_{\text{LE}} \left(\bm{R}_i, \bm{R} \right)} \left((\bm{R} - \bm{I})s+\bm{I}\right)^{-1} ds \bm{Y} \right).
        \end{aligned}
    \end{equation*}
    % where $\bm{I}$ is the identity matrix. 
    Solving $\riegrad G \left(\bm{R}\right) = \bm{0}$, we obtain
    \begin{equation}\label{eq:LEmediTH}
        \sum\limits^m_{i=1} \frac{\Log \bm{R}}{d_{\text{LE}} \left(\bm{R}_i, \bm{R} \right)} = \sum\limits^m_{i=1} \frac{ \Log\bm{R}_i}{d_{\text{LE}} \left(\bm{R}_i, \bm{R} \right)},
    \end{equation}
    which gives the fixed-point algorithm \eqref{eq: LEmedianfixpoint}.
\end{proof}

{\bf BW mean and median.}
The BW mean can be calculated by fixed-points algorithms (\eqref{eq: BWMeanfixpoint1} and \eqref{eq: BWMeanfixpoint2} below), whose convergence is known to be slow \cite{alvarez2016fixed}. 
To obtain the BW mean and median, we will introduce the Riemannian gradient descent algorithms which are relatively faster.

\begin{thm}\label{thm: BWmeanfix}
For $m$ HPD matrices $\{ \bm{R}_{i}\}_{i=1}^m$, the BW mean can be computed by either of the following fixed-point algorithms \cite{BHATIA2019165, alvarez2016fixed},
\begin{equation}  \label{eq: BWMeanfixpoint1}
    % \overline{\bm{R}}_{t+1} =  \left( \sum\limits_{j=1}^{m} \cfrac{1}{\left\{\delta_{\mathrm{TSL}}\left(\overline{\bm{R}}_{t},\bm{R}_j \right) \right\}^\frac{1}{2} \sqrt{1+\norm{\bm{R}_j}^2}}\right)^{-1}\left(\sum\limits_{i=1}^{m} \cfrac{ \bm{R}_i }{ \left\{\delta_{\mathrm{TSL}}\left(\overline{\bm{R}}_{t},\bm{R}_i \right) \right\}^\frac{1}{2} \sqrt{1+\norm{\bm{R}_i}^2}}\right)  \\
    % \bm{R}_{t+1} = \sum\limits_{i=1}^{m} \cfrac{ \left(\bm{R}_t^{\frac12} \bm{R}_i \bm{R}_t^{\frac12} \right)^{\frac12} }{ d_{BW}\left( \bm{R}_{i},\bm{R}_{t} \right) }, \\
    \bm{R}_{t+1} = \sum\limits_{i=1}^{m} \left(\bm{R}_t^{\frac12} \bm{R}_i \bm{R}_t^{\frac12} \right)^{\frac12}, \\
\end{equation}
and 
\begin{equation}  \label{eq: BWMeanfixpoint2}
    % \bm{R}_{t+1} = \bm{R}_t^{-\frac12} \left( \sum\limits_{i=1}^{m} \cfrac{ \left(\bm{R}_t^{\frac12} \bm{R}_i \bm{R}_t^{\frac12} \right)^{\frac12} }{ d_{\text{BW}}\left( \bm{R}_{i},\bm{R}_{t} \right) }\right)^2 \bm{R}_t^{-\frac12}.
    \bm{R}_{t+1} = \bm{R}_t^{-\frac12} \left( \sum\limits_{i=1}^{m} \left(\bm{R}_t^{\frac12} \bm{R}_i \bm{R}_t^{\frac12} \right)^{\frac12} \right)^2 \bm{R}_t^{-\frac12}.
\end{equation}
% \begin{equation}  \label{eq:WasMeanfixpoint2}
% % \overline{\bm{R}}_{t+1} =  \left( \sum\limits_{j=1}^{m} \cfrac{1}{\left\{\delta_{\mathrm{TSL}}\left(\overline{\bm{R}}_{t},\bm{R}_j \right) \right\}^\frac{1}{2} \sqrt{1+\norm{\bm{R}_j}^2}}\right)^{-1}\left(\sum\limits_{i=1}^{m} \cfrac{ \bm{R}_i }{ \left\{\delta_{\mathrm{TSL}}\left(\overline{\bm{R}}_{t},\bm{R}_i \right) \right\}^\frac{1}{2} \sqrt{1+\norm{\bm{R}_i}^2}}\right)  \\
%     \bm{R}_{t+1} = \bm{R}_t^{-\frac12} \left( \sum\limits_{i=1}^{m} \cfrac{ \left(\bm{R}_t^{\frac12} \bm{A}_i \bm{R}_t^{\frac12} \right)^{\frac12} }{ d_{BW}\left( \bm{R}_{t},\bm{A}_i \right) }\right)^2 \bm{R}_t^{-\frac12}.
% \end{equation}
\end{thm}
\begin{proof} A detailed proof seems not available in the literature and is provided here.
    Let $G(\bm{R})$ be the objective function for the BW mean which is given by 
    \begin{equation}
        G(\bm{R}) = \frac{1}{m} \sum_{i=1}^{m} d^2_{\text{BW}}\left(\bm{R}_i,\bm{R} \right).
    \end{equation}
    By using the Fr\'echet derivative \cite{LangSerge1995}, its Riemannian gradient satisfies
    % \begin{equation}
        \begin{align}
            & \inpror{ \riegrad G\left(\bm{R}\right), \bm{Y} } : = \frechetderi G\left(\bm{R} + \varepsilon \bm{Y}\right) \nonumber\\
            % & = \tr \left( \bm{Y} \right) - \frac{2}{m} \sum_{i=1}^{m} \frechetderi \tr \left\{\left( \bm{R}_i^{\frac12} \left( \bm{R}+ \epsilon \bm{Y} \right) \bm{R}_i^{\frac12}\right)^{\frac12}\right\} \\
            % & = \tr \left( \bm{Y} \right) - \frac{2}{m} \sum_{i=1}^{m} \tr \left(\frac12 \left( \bm{R}_i^{\frac12} \bm{R} \bm{R}_i^{\frac12}\right)^{-\frac12} \bm{R}_i^{\frac12} \bm{Y} \bm{R}_i^{\frac12} \right) \\
            & = \tr \left( \frac{1}{m} \sum\limits_{i=1}^{m} \left( \bm{I} - \bm{R}_i^{\frac12} \left( \bm{R}_i^{-\frac12} \bm{R}^{-1} \bm{R}_i^{-\frac12}\right)^{\frac12} \bm{R}_i^{\frac12}  \right)\bm{Y} \right) \nonumber \\ 
            & = \tr \left( \frac{1}{m} \sum\limits_{i=1}^{m} \left( \bm{I} - \bm{R}_i \# \bm{R}^{-1} \right) \bm{Y} \right), \label{eq:BWgrad}
            % & = \frac12 \tr \left( \Lyapunov{\bm{R}}{\riegrad G\left(\bm{R}\right)} \bm{Y}\right) ,
        \end{align}
    % \end{equation}
    where $\bm{Y} \in T_{\bm{R}} \mathscr{P}(N,\mathbb{C}).$ 
    % and $ \Lyapunov{\bm{R}}{\riegrad G\left(\bm{R} \right)} = 2 \sum\limits_{i=1}^{m} \left( \bm{I} - \bm{R}_i \# \bm{R}^{-1} \right). $
    Note that we used the fact 
    \begin{equation*}
        \tr \left( \frac{d}{d\varepsilon} \bm{P}^{\frac12}(\varepsilon) \right) = \frac12 \tr \left( \bm{P}^{-\frac12}(\varepsilon) \frac{d}{d\varepsilon} \bm{P}(\varepsilon) \right),
    \end{equation*}
    % \begin{equation*}
    %     \tr \left( \frechetderi \bm{P}^{\frac12}(\varepsilon) \right) = \frac12 \tr \left( \frechetderi  \bm{P}(\varepsilon)  \bm{P}^{-\frac12}(0) \right),
    % \end{equation*} 
    where $\bm{P}(\varepsilon) \in  \mathscr{P}(N,\mathbb{C})$. 
    From the definition of the BW metric \eqref{eq: BWmetric}, we have 
    \begin{equation}\label{eq:BWmetgrad}
        \inpror{ \riegrad G\left(\bm{R}\right), \bm{Y} } = \frac12 \tr \left( \Lyapunov{\bm{R}}{\riegrad G\left(\bm{R}\right)} \bm{Y}\right).
    \end{equation}
    Comparing \eqref{eq:BWgrad} and \eqref{eq:BWmetgrad}, we obtain
    \begin{equation}
        \Lyapunov{\bm{R}}{ \riegrad G\left(\bm{R}\right)} = \frac{2}{m} \sum\limits_{i=1}^{m}\left( \bm{I} - \bm{R}_i \# \bm{R}^{-1} \right).
    \end{equation}
    Consequently, from the definition of the Lyapunov operator, the Riemannian gradient is
    \begin{equation*} \label{eq:riegradBWMean}
        \begin{aligned}
            \riegrad G\left(\bm{R}\right) &= \bm{R} \Lyapunov{\bm{R}}{\riegrad G\left(\bm{R} \right)} + \Lyapunov{\bm{R}}{\riegrad G\left(\bm{R} \right)} \bm{R} \\
            & = \frac{2}{m} \sum\limits_{i=1}^{m} \bm{R} \left( \bm{I} - \bm{R}_i \# \bm{R}^{-1} \right) \\
            &\quad\quad + \frac{2}{m} \sum\limits_{i=1}^{m} \left( \bm{I} - \bm{R}_i \# \bm{R}^{-1} \right) \bm{R}.
            % & = 2 \sum\limits_{i=1}^{m} \bm{R} \left( \bm{I} - \bm{R}_i^{\frac12} \left( \bm{R}_i^{\frac12} \bm{R} \bm{R}_i^{\frac12}\right)^{-\frac12} \bm{R}_i^{\frac12}  \right) + 2 \sum\limits_{i=1}^{m} \left( \bm{I} - \bm{R}_i^{\frac12} \left( \bm{R}_i^{\frac12} \bm{R} \bm{R}_i^{\frac12}\right)^{-\frac12} \bm{R}_i^{\frac12}  \right) \bm{R}.
        \end{aligned}
    \end{equation*}
    Since the BW mean is the solution of $\riegrad G\left(\bm{R}\right) = \bm{0}$, we have 
    \begin{equation} \label{eq: Ap1}
        \sum\limits_{i=1}^{m} \bm{R} \left( \bm{I} - \bm{R}_i \# \bm{R}^{-1} \right) + \sum\limits_{i=1}^{m} \left( \bm{I} - \bm{R}_i \# \bm{R}^{-1} \right) \bm{R} = \bm{0}.
        % & = 2 \sum\limits_{i=1}^{m} \bm{R} \left( \bm{I} - \bm{R}_i^{\frac12} \left( \bm{R}_i^{\frac12} \bm{R} \bm{R}_i^{\frac12}\right)^{-\frac12} \bm{R}_i^{\frac12}  \right) + 2 \sum\limits_{i=1}^{m} \left( \bm{I} - \bm{R}_i^{\frac12} \left( \bm{R}_i^{\frac12} \bm{R} \bm{R}_i^{\frac12}\right)^{-\frac12} \bm{R}_i^{\frac12}  \right) \bm{R}.
    \end{equation}
    By the Appendix, equation \eqref{eq: Ap1} is equivalent to 
    \begin{equation}\label{eq: Ap2}
        \sum\limits_{i=1}^{m} \left( \bm{I} - \bm{R}_i \# \bm{R}^{-1} \right) = \bm{0}.
    \end{equation}
    By using the symmetry $\bm{R}_i \# \bm{R}^{-1} = \bm{R}^{-1} \# \bm{R}_i$, we obtain
    \begin{equation}
        \sum\limits_{i=1}^{m} \left( \bm{I} - \bm{R}^{-1} \# \bm{R}_i \right) = \bm{0},
    \end{equation}
    and consequently, 
    % \begin{equation} 
    %     \sum\limits_{i=1}^{m} \left( \bm{I} - \bm{R}^{-\frac12} \left(\bm{R}^{\frac12} \bm{R}_i \bm{R}^{\frac12} \right)^{\frac12} \bm{R}^{-\frac12} \right) = \bm{0}, 
    % \end{equation}
    \begin{equation} \label{eq: BWmean_solution}
        \bm{R} = \sum\limits_{i=1}^{m} \left(\bm{R}^{\frac12} \bm{R}_i \bm{R}^{\frac12} \right)^{\frac12},
    \end{equation}
    which yields the iterations \eqref{eq: BWMeanfixpoint1} and \eqref{eq: BWMeanfixpoint2}.
\end{proof}

\begin{thm}\label{prop: BWm}
    For $m$ HPD matrices $\{ \bm{R}_{i}\}_{i=1}^m$, the BW mean and median can respectively be computed by the Riemannian gradient descent algorithms
    \begin{equation} \label{eq: RGDofBWmean}
        \bm{R}_{t+1} = \left( \sum\limits_{i=1}^{m}  \left( \bm{R}_i \# \bm{R}_{t}^{-1} \right) \right) \bm{R}_{t} \left( \sum\limits_{i=1}^{m} \left( \bm{R}_i \# \bm{R}_{t}^{-1} \right) \right),
    \end{equation}
    %The BW median of the HPD matrices $\{ \bm{R}_{i}\}_{i=1}^m$ can be derived by the following Riemannian gradient descent algorithm,
    and 
    \begin{equation} \label{eq: RGDofBWmedian}
        \bm{R}_{t+1} = \bm{S}_t \bm{R}_{t}  \bm{S}_t ,
    \end{equation}
    where
    \begin{equation}
        \bm{S}_t = \bm{I} - \eta_t \sum\limits_{i=1}^{m} \cfrac{ \bm{I} - \bm{R}_i \# \bm{R}_t^{-1} }{ d_{\text{BW}}\left( \bm{R}_{i},\bm{R}_{t} \right) }.
    \end{equation}
    These are consistent with previous studies for probability spaces \cite{altschuler2021averaging}. 
\end{thm}
\begin{proof}
Using the exponential map \eqref{eq: ExpBW}, the Riemannian gradient descent algorithm for the BW mean is 
    \begin{equation*}
        \begin{aligned}
            \bm{R}_{t+1} & = \Exp\left(\bm{R}_{t},-\eta_t \riegrad G\left(\bm{R}_t\right)\right) \\
            & = \left(\bm{I} + \Lyapunov{\bm{R}_t}{- \eta_t \riegrad G\left(\bm{R}_t\right)} \right) \bm{R}_t \\ 
            & ~~~~~~~~~~~ \times\left(\bm{I} + \Lyapunov{\bm{R}_t}{- \eta_t \riegrad G\left(\bm{R}_t\right)} \right),
        \end{aligned}
    \end{equation*}
    which is exactly \eqref{eq: RGDofBWmean}. Proof for the BW median is similar. 
\end{proof}

\subsection{Computational Complexity}
\label{sec: CC}
In this subsection, we study computational complexity of algorithms for the AIRM, LE, BW means and medians given by Theorems \ref{prop: AIRMm}, \ref{th: LEmean}, \ref{prop: LEmd}, \ref{prop: BWm} and the arithmetic mean \eqref{eq:arith}. 
For simplicity, we focus solely on the leading terms and only provide the computational complexity for a single step of numerical iterations.
The complexity figures assume that $m$ number of $N\times N$ HPD matrices are given and the arithmetic operations with individual elements have a complexity of $O(1)$. 
The following facts are used: $\bm{R}^{-1} \sim O(N^3)$, $\bm{R}^{\frac{1}{2}} \sim O(N^3)$, and $\operatorname{Log}\bm{R} \sim O(N^4)$ \cite{Hig2008, 9764734}. 
The matrix exponential in all algorithms only deals with Hermitian matrices, and one way to calculate the matrix exponential is through eigenvalue decomposition, whose complexity is $O(N^3)$, 
the same as that of matrix inversion and matrix square root.

\begin{table}[ht]
\caption{Computational complexity of the means and medians.}
\centering
\begin{tabular}{|c|c|}
        \hline
        Geometric measures  & Complexity  \\
        \hline
        Arithmetic mean \eqref{eq:arith} &  $O(N^2(m-1))$ \\
        \hline
        AIRM mean(Equation \eqref{eq: gradAIRMmean}, per iteration)  &  $O(N^4m)$  \\
        \hline
        AIRM median(Equation \eqref{eq: gradAIRMmedian}, per iteration)  &  $O(N^4m)$  \\
        \hline
        LE mean (Theorem \ref{th: LEmean})&  $ O(N^4m) $ \\
        \hline
        LE median (Theorem \ref{prop: LEmd}, per iteration)&  $ O(N^4m) $ \\
        \hline
        BW mean (Equation \eqref{eq: RGDofBWmean}, per iteration) & $O(N^3(m+1))$  \\
        \hline
        BW median (Equation \eqref{eq: RGDofBWmedian}, per iteration) & $O(N^3(m+1))$  \\
        \hline
    \end{tabular}
    \label{tab: m}
\end{table}

Among the numerical methods in Table \ref{tab: m}, the computational complexity of the BW metric is less than that of the other metrics.
Next, we investigate the computational cost and the convergence rate of the fixed-point algorithms \eqref{eq: BWMeanfixpoint1}, \eqref{eq: BWMeanfixpoint2} and the Riemannian gradient descent algorithm  \eqref{eq: RGDofBWmean} for computing the BW mean. In each algorithm, the BW mean of 10 HPD matrices is computed numerically:
%We prepare the 10 HPD matrices and compute the BW mean by each algorithm. 
the initial point is chosen as the arithmetic mean of the prepared matrices and the tolerance $|\Delta \bm{R}|$ is $10^{-5}$. 
According to \Figref{fig: convergencerate} and \Tabref{tab: time}, it is obvious that the Riemannian gradient descent algorithm \eqref{eq: RGDofBWmean} is faster than the fixed-point algorithms \eqref{eq: BWMeanfixpoint1} and \eqref{eq: BWMeanfixpoint2}.

\begin{table}[ht]
    \caption{Computational cost about the BW mean.}
    \centering
    \begin{tabular}{|c|c|}
            \hline
            Methods  & Time (second)  \\
            \hline
            The fixed-point algorithm \eqref{eq: BWMeanfixpoint1} &  0.997728  \\
            \hline
            The fixed-point algorithm \eqref{eq: BWMeanfixpoint2} &  0.080401 \\
            \hline
            The Riemannian gradient descent algorithm \eqref{eq: RGDofBWmean} &  0.052697 \\
            \hline
        \end{tabular}
        \label{tab: time}
    \end{table}

\begin{figure}[htbp]
	\centering
	\includegraphics[width = \linewidth]{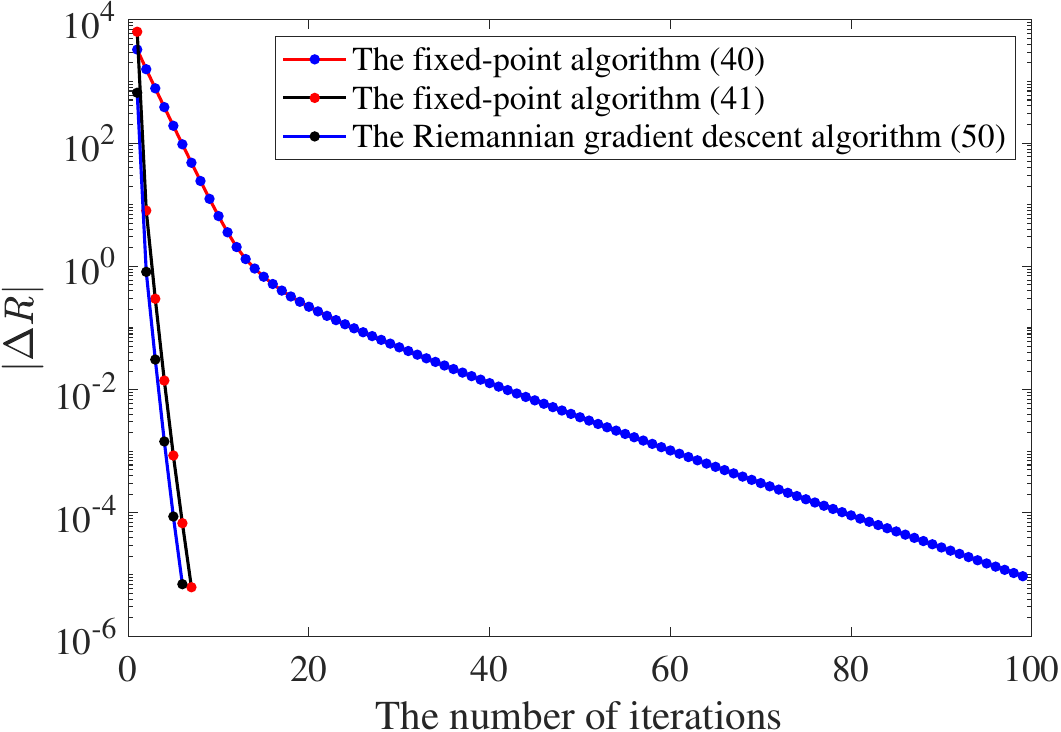}
	% \captionsetup{width=0.8\textwidth}
	\caption{Convergence rates of the fixed-point algorithms \eqref{eq: BWMeanfixpoint1}, \eqref{eq: BWMeanfixpoint2} and the Riemannian gradient descent algorithm \eqref{eq: RGDofBWmean}}
    \label{fig: convergencerate}
\end{figure}

In the matrix-CFAR, we can leverage the AIRM (or LE, BW) mean and median as the matrix $\bm{R}_g$ in \eqref{eq:decison_of_R_g} and \Figref{fig: processMCFAR}.
From numerical simulations in the following section, the Riemannian gradient descent algorithms of Theorems \ref{prop: AIRMm} and \ref{prop: BWm} and the fixed-point algorithms of Theorems \ref{th: LEmean} and \ref{prop: LEmd} are all convergent 
when the initial point is set as the arithmetic mean and the tolerance is set to $10^{-3}$.

\section{NUMERICAL SIMULATIONS}
\label{sec: DP}
To access the detection performance, numerical simulations are performed. 
As a comparison, we present the detection performances of the Riemannian geometric means and medians under the AIRM, the LE metric, and the BW metric, as well as the AMF, ANMF and CA-CFAR methods in two cases. 
In the first case, we treat the ideal steering vector as the target signal, while in the second case, we account for the presence of steering signal mismatches.

\subsection{Scenario with the Ideal Steering Vector as the Target Signal}
In this subsection, 
we investigate the scenario using the ideal steering vector
\begin{equation}\label{eq: target}
    \dfrac{1}{\sqrt{N}}\left(1, \exp(-\operatorname{i}2\pi f_d), \ldots, \exp(-\operatorname{i}2\pi f_d(N-1))\right)^{\operatorname{T}}, 
\end{equation}
as the target $\bm{s}$.
Here $\operatorname{i}$ is the imaginary unit and $f_d $ is the normalized Doppler frequency of target.
% The other is the vector considering the presence of steering vector mismatches, which is defined by 
% \begin{equation}\label{eq: steering vector mismatche}
%     \bm{s} = \dfrac{1}{\sqrt{N}}\exp( \operatorname{i} 2\pi) \left( \bm{e}_1 \cos \theta_{\rm{mis}} + \frac{\bm{y}}{\norm{\bm{y}}}\sin \theta_{\rm{mis}} \right),
% \end{equation}
% where $\bm{e}_1 = \left(1,0,\ldots,0\right)^{\trans}$, $\norm{\cdot}$ denotes the Euclidean norm, the vector $\bm{y}$ is an $N$-dimensional vector, orthogonal to $\bm{e}_1$, 
% whose components are independent and identically distributed from $CN\left(\bm{0},\bm{I} \right)$, and the parameter $\theta_{\rm{mis}}$ characterises the steering mismatch 
% and denotes the angular difference or cone angle between the nominal and actual spatial signatures \cite{1561887,liu2022multichannel}. In other words, it conforms to the following relationship:
% \begin{equation}
%     \cos^2 \theta_{\rm{mis}} = \frac{\left| \bm{s}^{\hermconj} \bm{e}_1  \right|^2}{\norm{\bm{s}}^2\norm{\bm{e}_1}^2}.
% \end{equation}
The clutter is considered as the compound-Gaussian clutter, which has been commonly used for modelling heavy-tailed sea clutter \cite{seaclutter} and is defined by
\begin{equation}
	\bm{c} = \sqrt{\tau} \bm{z}, 
\end{equation}
where $\bm{z}$ is fast fluctuating and called the speckle component.
The model of the speckle component is the circularly-symmetric Gaussian distribution $ CN \left( \bm{0}, \bm{\Sigma} \right)$ with zero-mean and a covariance matrix $\bm{\Sigma}$, 
the probability density function (PDF) of which is 
\begin{equation}
	p\left(\bm{z}\ |\ \bm{0}, \bm{\Sigma}\right) = \frac{1}{\pi^N \det\left(\bm{\Sigma}\right)} \exp \left(- \bm{z}\hermconj \bm{\Sigma}^{-1} \bm{z} \right).
\end{equation}
The random variable $\tau$ is relatively slowly varying and called the texture component. 
In this paper, the model of the texture component $\tau$ is a gamma distribution of a shape parameter $ \alpha$ and a scale parameter $\beta$.
Its PDF is 
\begin{equation}
	q\left(\tau\ |\ \alpha, \beta\right) = \frac{1}{\beta^\alpha \Gamma\left(\alpha\right)} \tau^{\alpha -1} \exp \left(\frac{- \tau}{\beta} \right).
\end{equation}
Here $\Gamma (\cdot)$ is the gamma function. 

As a consequence, the clutter model is the complex K-distribution, or the compound-Gaussian model \cite{seaclutter}. 
In the simulations, the speckle component is sampled from $CN\left(\bm{0},\bm{\Sigma} \right)$ of which the known covariance matrix $\bm{\Sigma}$ is 
\begin{equation}\label{eq: sig}
	\bm{\Sigma} = \bm{\Sigma}_0 + \bm{I},
\end{equation}
where
% Here $\bm{I}_N$ is the $N$-dimensional identity matrix and
\begin{equation}
	\bm{\Sigma}_0 \left(i,j\right) = \sigma_c^2 \rho^{\left| i-j \right|} \exp \left( \mathrm{i} 2\pi f_c (i-j) \right), \nonumber \ \ i,j = 1,2, \ldots , N.
\end{equation}
%Here, $\bm{I}_N$ is the $N$ dimensional identity matrix, $\sigma_c$ is the clutter-to-noise ratio, $\rho$ is the one-lag coefficient, and $f_c$ is the normalized Doppler frequency of the clutter. 
%In the simulation, we set $\rho =0.9 ,\ \sigma_c^2 = 20 \mathrm{dB},\ f_c = 0.2$, the shape parameter $ \alpha = 4 $ and the scale parameter $\beta = 3$.
Here the clutter-to-noise ratio $\sigma_c^2$ is $20 \mathrm{dB}$, the one-lag coefficient $\rho$ is 0.9, and the normalized Doppler frequency of clutter $f_c$ is 0.2. 
Additionally, the shape parameter, $\alpha $ is 4 and the scale parameter, $\beta$ is 3.
The signal-to-clutter ratio (SCR) is 
\begin{equation}\label{eq: SCR}
    \mathrm{SCR} = |a|^2 \bm{s}\hermconj \bm{R}^{-1} \bm{s},
\end{equation}
where $a$ is the amplitude coefficient in \eqref{eq: H01}, $\bm{s}$ is a target signal and $\bm{R}$ is the clutter autocovariance matrix.
Furthermore, the false alarm rate $P_{fa}$ is $10^{-3}$.
The dimension of the observation data, $N$, is 8, and we study the scenario when the numbers of observation data $m$ are $N,2N,3N$.
Two interferences of which normalized Doppler frequency is $f_i = 0.2$ are inserted into clutter,
the detection probability $P_d$ is computed through 2000 independent trials, and the threshold $\gamma$ is computed by $100/P_{fa}$.

% {\color{red} In the case of the ideal steering vector,}
\Figref{fig: DPforCG} demonstrates that the detection performance is enhanced with an increasing amount of observation data.
The AMF and the ANMF are invalid when $m$ is small, while the matrix-CFAR detectors perform effectively.
From \Figref{fig: DPforCG}, the BW mean and median we have proposed have the best detection performance when $m=N$. 
In $m \ge 2N$, the ANMF outperforms all other detectors. The reason is that the accuracy of the SCM is sufficient when $m \ge 2N$ \cite{SCMbyMLE}. 
The matrix-CFAR using the BW metric behaves better than that of the other Riemannian metrics.
In general, from these simulation results, it can be noticed that the matrix-CFAR outperforms the AMF and ANMF when the number of data available is limited.
This can be attributed to the difference in the effect of clutter between the matrix-CFAR and the AMF and ANMF.
Specifically, the AMF and ANMF use contaminated observation data directly, whereas the matrix-CFAR utilises the autocovariance matrix of the observation data. 
Therefore, we conduct the next simulation, which is a detection simulation of the AMF using the Riemannian geometric means and medians instead of the SCM. 
%In this simulation, we investigate the detection performance by replacing the SCM with the Riemannian geometric means and medians in the AMF.
We use the Gaussian clutter generated by $CN \left( \bm{0}, \bm{\Sigma} \right)$ and the clutter covariance matrix $\bm{\Sigma}$ given by \eqref{eq: sig} as the benchmark. Other parameters are the same.
As shown by \Figref{fig:compareAMF}, the AMFs utilising Riemannian geometric means and medians yield almost the same detection performance regardless of the number of data and better performance than the AMF via the SCM.

\begin{figure}[htbp]
    \begin{minipage}{\hsize}
        \centering
        \includegraphics[width=\hsize]{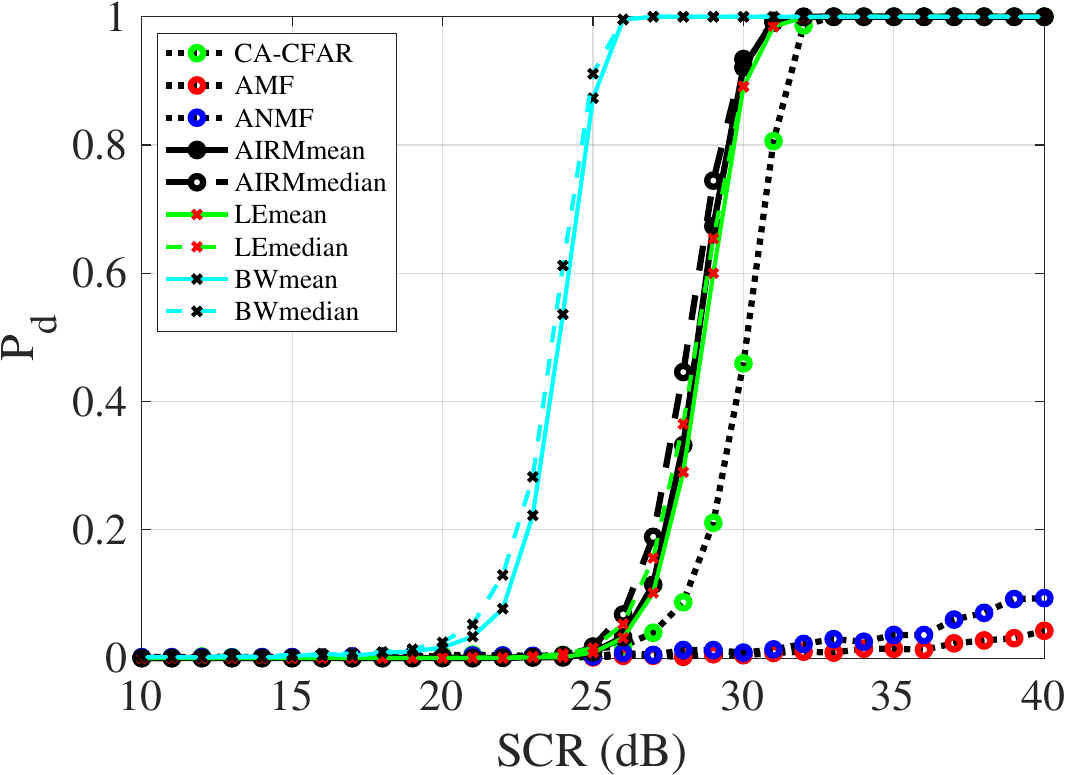}
        \subcaption{$m=N$}\label{fig:m=8gamma}
    \end{minipage}
  % \begin{minipage}{0.49\hsize}
  %     \centering
  %     \includegraphics[width=\hsize]{m=12gamma}
  %     \subcaption{Compound-Gaussian clutter, $m=12$}\label{fig:m=12gamma}
  % \end{minipage}
%%%%%%%%%%%%%%%%%%%%%%%%%%%%%%%%%%%%%%%%%%%%%%%%%%%%%%
    \begin{minipage}{\hsize}
        \centering
        \includegraphics[width=\hsize]{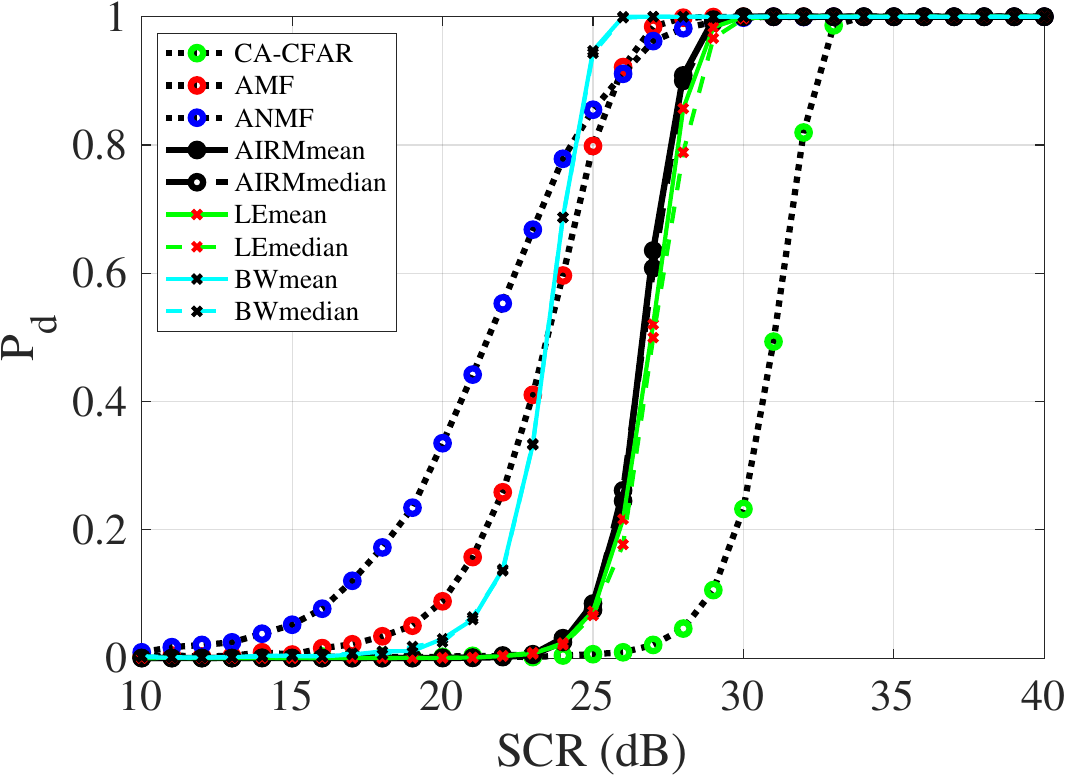}
        \subcaption{$m=2N$}\label{fig:m=16gamma}
    \end{minipage}
    \begin{minipage}{\hsize}
        \centering
        \includegraphics[width=\hsize]{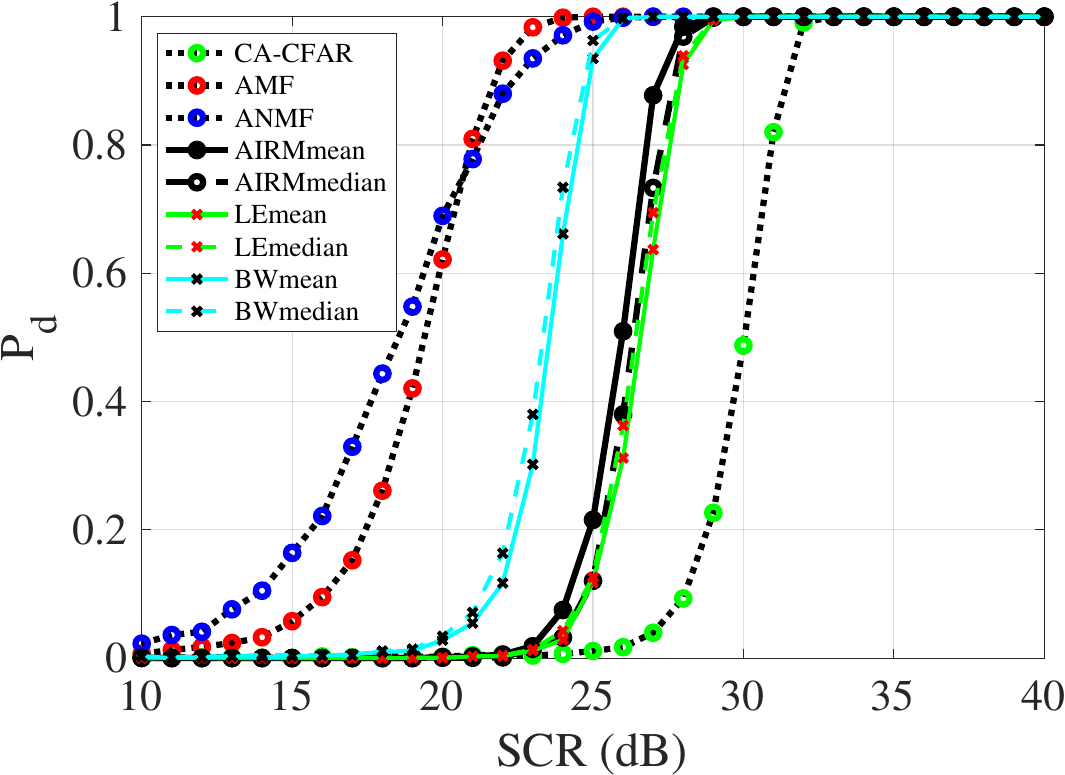}
        \subcaption{$m=3N$}\label{fig:m=24gamma}
    \end{minipage}
    % \captionsetup{width=0.8\textwidth}
    \caption{SCR versus $P_d$ in the cases of $m=N, 2N, 3N$ and  $f_d = 0.2$ in the scenario with the ideal steering vector as the target.}
    \label{fig: DPforCG}
\end{figure}

\begin{figure}[htbp]
    \begin{minipage}{\hsize}
        \centering
        \includegraphics[width=\hsize]{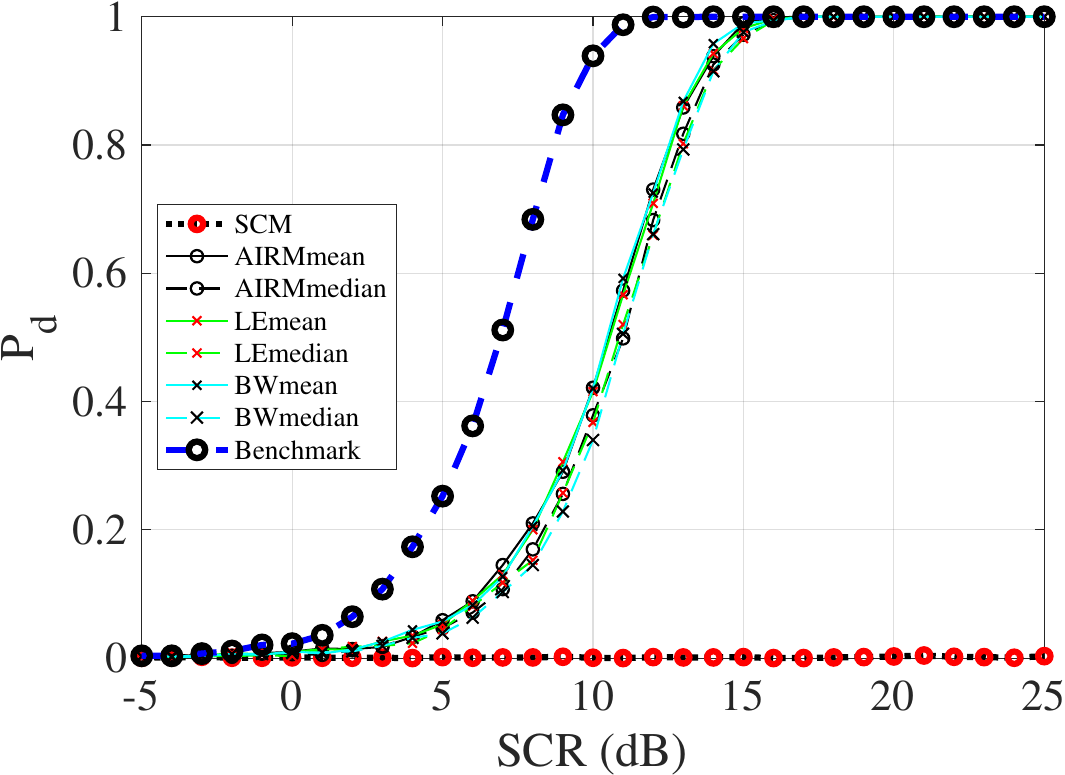}
        \subcaption{Gaussian clutter, $m=N$}\label{fig:m=8AMF}
    \end{minipage}
  % \begin{minipage}{0.49\hsize}
  %     \centering
  %     \includegraphics[width=\hsize]{m=12AMFgaussian}
  %     \subcaption{Gaussian clutter, $m=12$}\label{fig:m=12AMF}
  % \end{minipage}
%%%%%%%%%%%%%%%%%%%%%%%%%%%%%%%%%%%%%%%%%%%%%%%%%%%%%%
    \begin{minipage}{\hsize}
        \centering
        \includegraphics[width=\hsize]{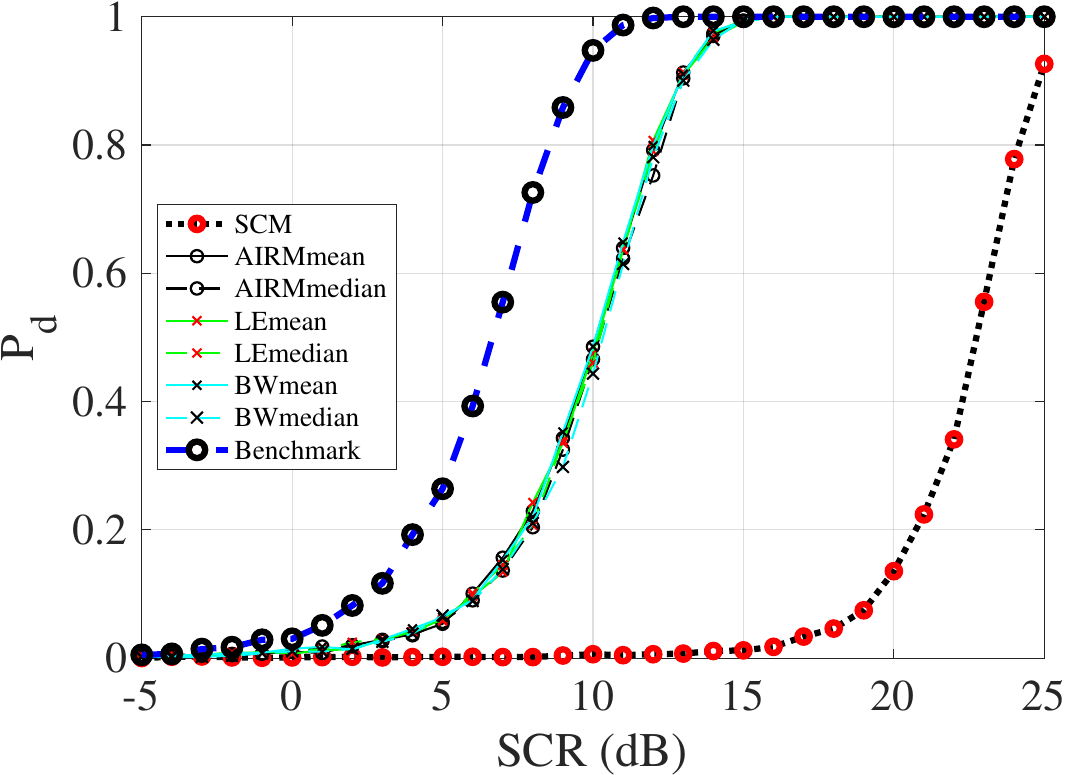}
        \subcaption{Gaussian clutter, $m=2N$}\label{fig:m=16AMF}
    \end{minipage}
    \begin{minipage}{\hsize}
        \centering
        \includegraphics[width=\hsize]{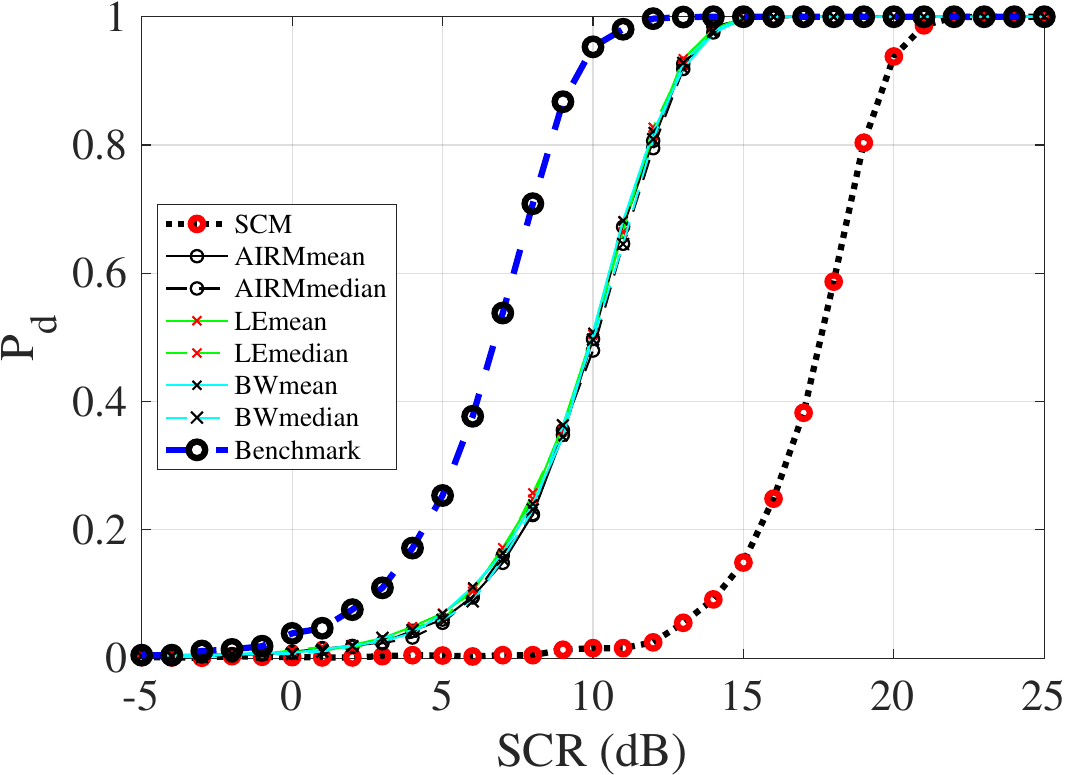}
        \subcaption{Gaussian clutter, $m=3N$}\label{fig:m=24AMF}
    \end{minipage}
\caption{SCR  versus $P_d$ in AMF replacing SCM by Riemannian geometric means and medians.}\label{fig:compareAMF}
\end{figure}

 Furthermore, to show the influence of  targets with different normalized Doppler frequencies $f_d$ on detection performance, the following simulations are conducted in the cases that clutter is modelled as the compound-Gaussian clutter, the SCR is $25$dB, $f_d$ varies in $[0,1]$, while the other parameters are kept as the same.
 \Figref{fig: PdvsFd} illustrates that it is relatively more difficult to detect the target with $f_d $ closer to the normalized Doppler frequency of clutter $f_c= 0.2$ (noticing the periodicity of frequency).  
In the matrix-CFAR, the BW distance can capture the difference between the target and clutter better compared with other distance functions.

\begin{figure}[htbp]
    \begin{minipage}{\hsize}
        \centering
        \includegraphics[width=\hsize]{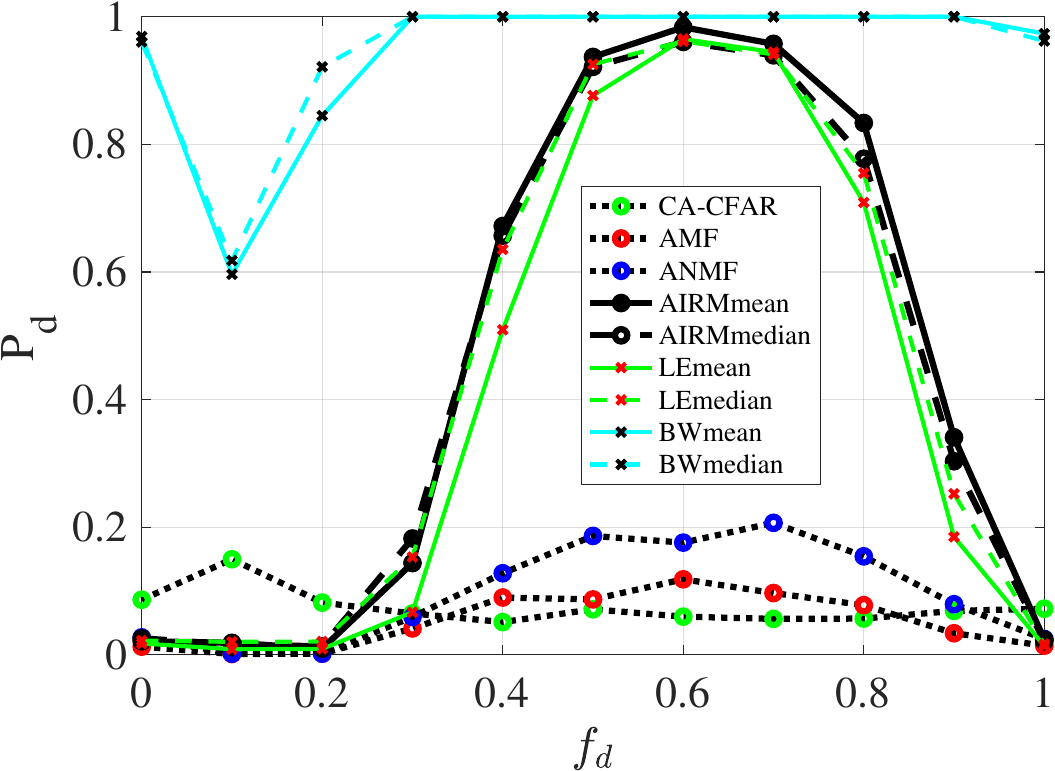}
        \subcaption{$m=N$}\label{fig:m=8gammaPdvsFd}
    \end{minipage}
  % \begin{minipage}{0.49\hsize}
  %     \centering
  %     \includegraphics[width=\hsize]{m=12gamma}
  %     \subcaption{Compound-Gaussian clutter, $m=12$}\label{fig:m=12gamma}
  % \end{minipage}
%%%%%%%%%%%%%%%%%%%%%%%%%%%%%%%%%%%%%%%%%%%%%%%%%%%%%%
    \begin{minipage}{\hsize}
        \centering
        \includegraphics[width=\hsize]{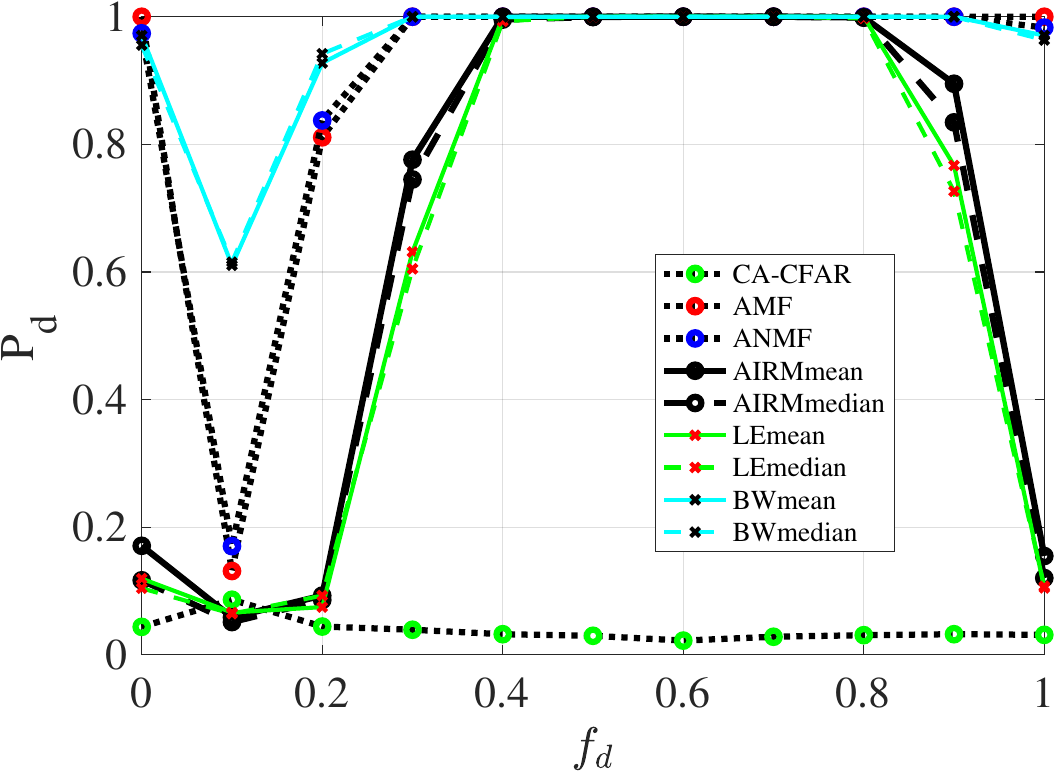}
        \subcaption{$m=2N$}\label{fig:m=16gammaPdvsFd}
    \end{minipage}
    \begin{minipage}{\hsize}
        \centering
        \includegraphics[width=\hsize]{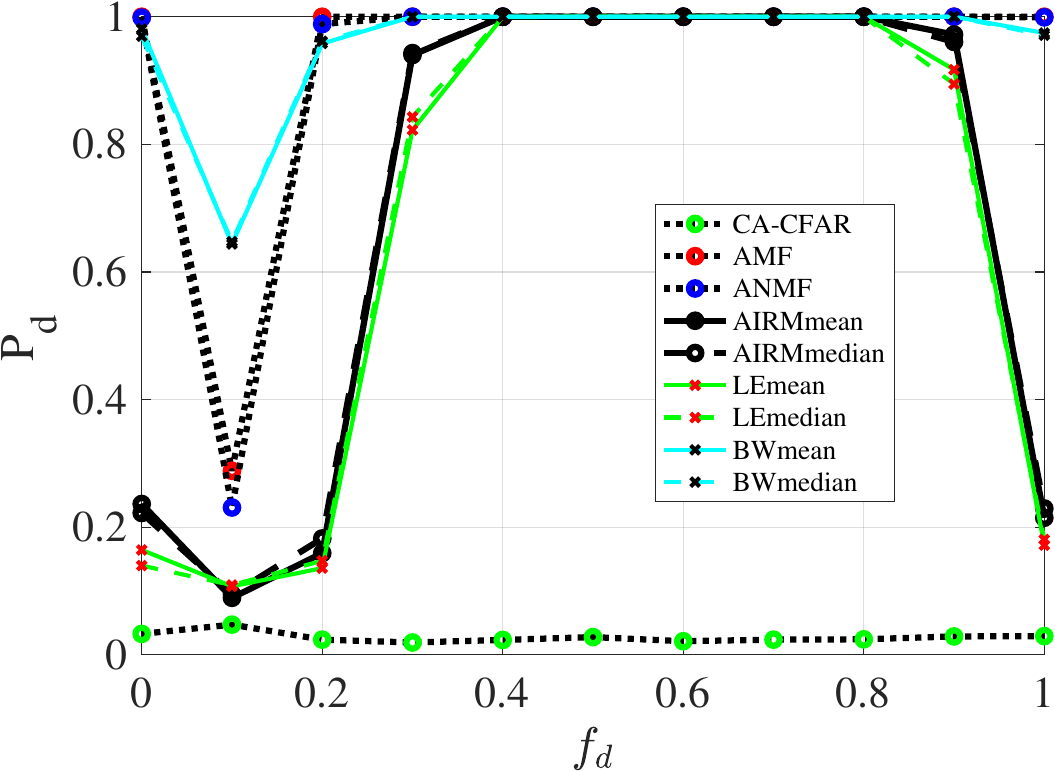}
        \subcaption{$m=3N$}\label{fig:m=24gammaPdvsFd}
    \end{minipage}
    % \captionsetup{width=0.8\textwidth}
    \caption{ Normalized Doppler frequency $f_d$ of target versus $P_d$ in the cases of $m=N, 2N, 3N $.  The SCR is $25$dB.}
    \label{fig: PdvsFd}
\end{figure}

% クラッタの周波数が0.2であるので、そこに近い周波数では検知が難しい、検出確率が0.1で最も小さくなるのは、compound gaussianで少し違うモデルを使っているからと思われる。
\subsection{Signal Mismatched Scenarios}
In this subsection, signal mismatched scenarios are taken into account. In the simulations, the target signal is given by
\begin{equation}\label{eq: steering vector mismatched}
    \bm{s} = \dfrac{1}{\sqrt{N}}\exp( \operatorname{i} 2\pi) \left( \bm{e}_1 \cos \theta_{\rm{mis}} + \frac{\bm{y}}{\norm{\bm{y}}}\sin \theta_{\rm{mis}} \right),
\end{equation}
where $\bm{e}_1 = \left(1,0,\ldots,0\right)^{\trans}$, $\norm{\cdot}$ denotes the Euclidean norm, and  $\bm{y}$ is an $N$-dimensional vector, orthogonal to $\bm{e}_1$, 
whose components are independent and identically distributed from $CN\left(\bm{0},\bm{I} \right)$. 
The parameter $\theta_{\rm{mis}}$ characterises the steering mismatch 
and denotes the angular difference or cone angle between the nominal and actual spatial signatures \cite{1561887,liu2022multichannel}. 
In other words, it conforms to the following relationship:
\begin{equation}
    \cos^2 \theta_{\rm{mis}} = \frac{\left| \bm{s}^{\hermconj} \bm{e}_1  \right|^2}{\norm{\bm{s}}^2\norm{\bm{e}_1}^2}.
\end{equation}
The cases of signal mismatched scenarios with $\theta_{\rm{mis}} = 1, 15, 30$ are investigated for $m = 3N$. 
The other parameters are the same as in the case when the ideal steering vector is treated as the target.

\Figref{fig: DPformismatch} shows that the detection performance of the matrix-CFAR is robust with respect to mismatched signals.
On the contrary, ANMF cannot detect the mismatched signal even in large SCR.
The detection performance of the AMF is better than that of the ANMF. Robustness of the AMF about mismatched signals is acceptable, but the higher $\theta_{\rm{mis}}$ is the worse the detection performance becomes.

\begin{figure}[htbp]
    \begin{minipage}{\hsize}
        \centering
        \includegraphics[width=\hsize]{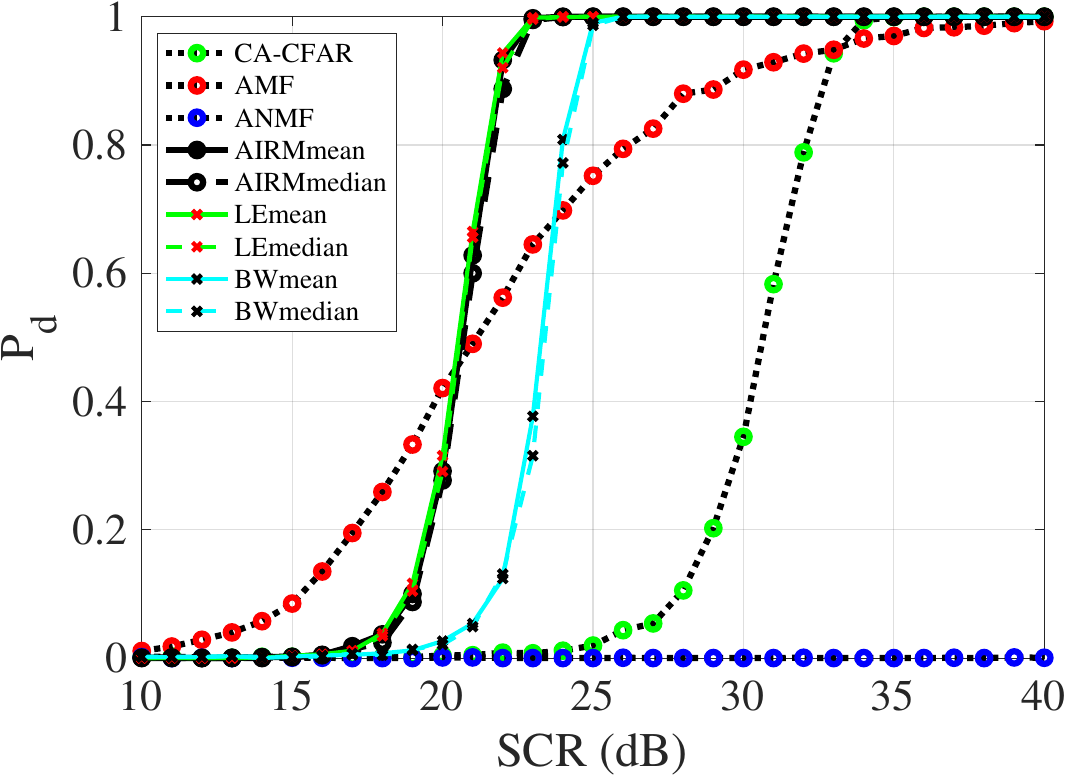}
        \subcaption{$\theta_{\rm{mis}}=1$}\label{fig: theta=1_m=24gamma}
    \end{minipage}
  % \begin{minipage}{0.49\hsize}
  %     \centering
  %     \includegraphics[width=\hsize]{m=12gamma}
  %     \subcaption{Compound-Gaussian clutter, $m=12$}\label{fig:m=12gamma}
  % \end{minipage}
%%%%%%%%%%%%%%%%%%%%%%%%%%%%%%%%%%%%%%%%%%%%%%%%%%%%%%
    \begin{minipage}{\hsize}
        \centering
        \includegraphics[width=\hsize]{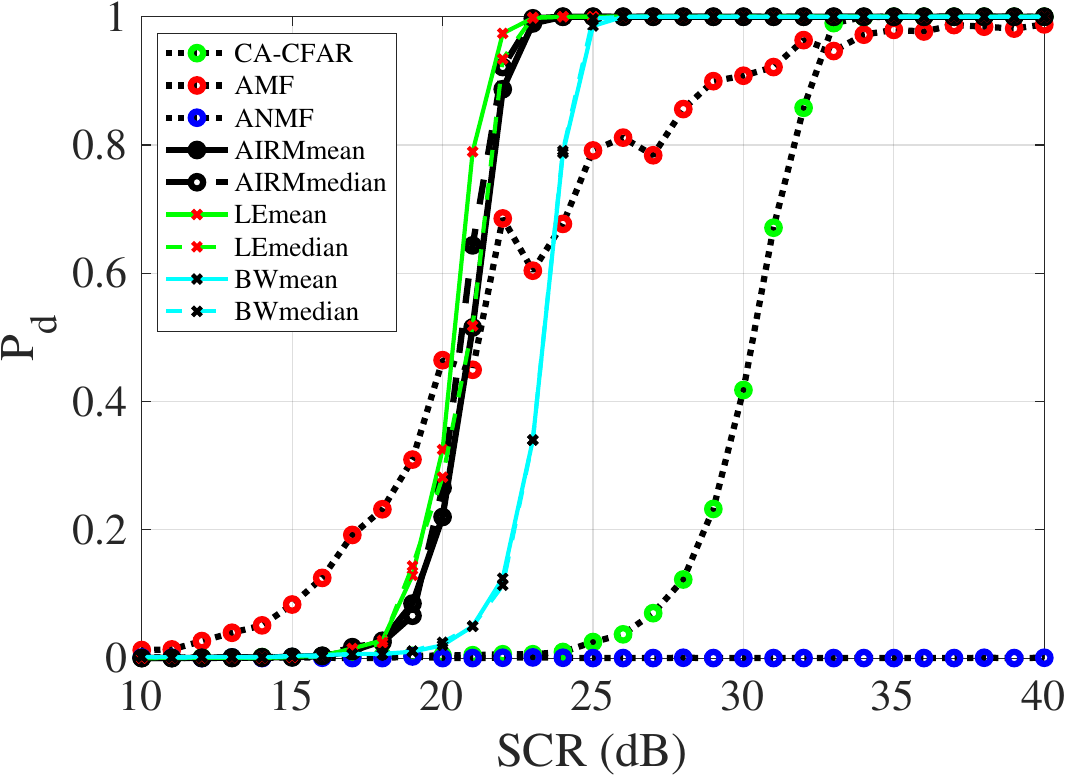}
        \subcaption{$\theta_{\rm{mis}}=15$}\label{fig: theta=15_m=24gamma}
    \end{minipage}
    \begin{minipage}{\hsize}
        \centering
        \includegraphics[width=\hsize]{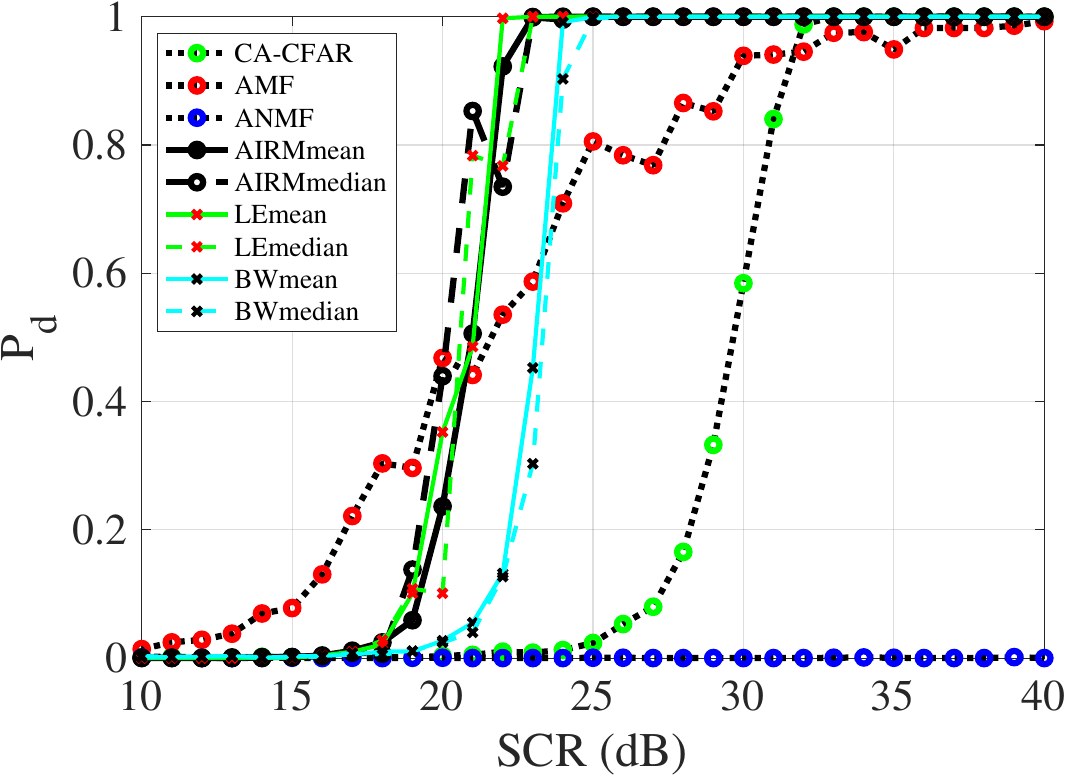}
        \subcaption{$\theta_{\rm{mis}}=30$}\label{fig: theta=30_m=24gamma}
    \end{minipage}
    % \captionsetup{width=0.8\textwidth}
    \caption{SCR  versus $P_d$ in the case of $m=3N$ in signal mismatched scenarios.}
    \label{fig: DPformismatch}
\end{figure}

\section{ROBUST ANALYSIS THROUGH INFLUENCE FUNCTIONS}
\label{sec: RA}
In the application to real problems, the observation data may also contain outliers, thus robustness about the estimator of the autocovariance matrix is of importance. 
The robust analysis against outliers through the influence function has been commonly used \cite{myarticle,10057006}. 
In this section, we review the definition of the influence function and conduct the robustness analysis. 

Let $\overline{\bm{R}}$ be the Riemannian geometric mean or median of given $m$ HPD matrices $\{ \bm{R}_{i} \}_{i=1}^m$, and let $\widehat{\bm{R}}$ be the Riemannian geometric mean or median of the contaminated HPD matrices containing $n$ outliers $\{ \bm{P}_j\}_{j=1}^n$. 
%For a given set of $n$ HPD matrices $\{ \bm{R}_i\}_{i=1}^n$ with mean or median $\overline{\bm{R}}$, 
By considering the outliers as a perturbation $\varepsilon \left( \varepsilon \ll 1 \right)$ of the matrix $\overline{\bm{R}}$, $\widehat{\bm{R}}$ can be rewritten as 
\begin{equation}
	\widehat{\bm{R}} = \overline{\bm{R}} + \varepsilon \bm{H} \left( \{ \bm{R}_i\}_{i=1}^m,\{ \bm{P}_j\}_{j=1}^n \right) + \mathcal{O} \left( \varepsilon^2\right),
\end{equation}
where $ \bm{H} \left( \{ \bm{R}_{i}\}_{i=1}^m,\{ \bm{P}_j\}_{j=1}^n \right) $ is a Hermitian matrix and depends on the given HPD matrices $\{ \bm{R}_{i}\}_{i=1}^m$ and the outliers $\{ \bm{P}_j\}_{j=1}^n$.
We define the function 
\begin{equation*}\label{eq: inff}
    f\left( \{ \bm{R}_i\}_{i=1}^m,\{ \bm{P}_j\}_{j=1}^n \right) \coloneqq\dfrac{\normf{\bm{H}\left( \{ \bm{R}_i\}_{i=1}^m,\{ \bm{P}_j\}_{j=1}^n \right)}}{\normf{\overline{\bm{R}}}} 
\end{equation*}
as the {\it influence function}.

We consider the mean case here and the median case is similar. 
To compute the corresponding influence functions, the objective function $F(\bm{R})$ of $\{ \bm{R}_{i}\}_{i=1}^m$ and $\{ \bm{P}_j\}_{j=1}^n$ is defined by
\begin{equation*}
	F(\bm{R}):= (1-\varepsilon)\frac{1}{m}\sum_{i=1}^{m} d^2\left(\bm{R}_i, \bm{R} \right) + \varepsilon\frac{1}{n}\sum_{j=1}^{n} d^2\left(\bm{P}_j,\bm{R} \right),
\end{equation*}
where $d$ is the distance function for each mean. For medians, the discrete `integrand' is replaced by $d$.
Since $\bm{\widehat{R}}$ is the solution of $\riegrad F(\bm{R}) = \bm{0}$, we obtain 
% \footnote{Note that in the Riemannian case, this should be $\operatorname{grad} G(\widehat{\bm{R}})=0$, which is nevertheless equivalent to $\nabla G(\widehat{\bm{R}})=0$ for HPD manifolds equipped with the AIRM.} 
\begin{equation}
    \begin{aligned}
		& (1-\varepsilon) \frac{1}{m}\sum_{i=1}^{m} \riegrad d^2\left(\bm{R}_i, \bm{\widehat{R}} \right) \\ 
        & ~~~~~ +  \varepsilon\frac{1}{n}\sum_{j=1}^{n} \riegrad d^2\left(\bm{P}_j, \bm{\widehat{R}} \right) = \bm{0}.
    \end{aligned}
		%(1-\varepsilon) \frac{1}{n}\sum_{i=1}^{n} \frac{ \widehat{\bm{R}} \Log \left(\bm{R}_i^{-1} \widehat{\bm{R}} \right)}{ d_R\left(\widehat{\bm{R}},\bm{R}_i \right) } +  \varepsilon\frac{1}{n}\sum_{j=1}^{n} \frac{ \widehat{\bm{R}} \Log \left(\bm{P}_j^{-1} \widehat{\bm{R}} \right) }{ d_R\left(\bm{\widehat{R}},\bm{P}_j \right)} = 0 ,\\
		%\hspace{-5mm}
\end{equation}
%where
%\begin{equation}
% \nabla d\left(\bm{\widehat{R}},\bm{R}_i \right) = \left(\nabla d\left(\bm{R},\bm{R}_i \right)\right)_{\bm{R}=\widehat{\bm{R}}}, 
%\end{equation}
%and so forth.
Differentiating $\riegrad F(\widehat{\bm{R}})=\bm{0}$ about $\varepsilon$ at $\varepsilon=0$ , a matrix equation with respect to $\bm{H}$ is derived by 
\begin{equation}\label{eq:inf0}
    \begin{aligned}
        & \frac{1}{m}\frechetderi \left( \sum_{i=1}^{m} \riegrad d^2\left(\bm{R}_i,\bm{\widehat{R}} \right) \right) \\ 
        & ~~~~~~ + \frac{1}{n}\sum_{j=1}^{n} \riegrad d^2\left(\bm{P}_j,\bm{\widehat{R}} \right) = \bm{0} .
    \end{aligned}
\end{equation}
% By using an orthonormal basis of Hermitian matrices, the matrix equation \eqref{eq:inf0} can be formulated as a linear system equation that can be solved easily.
An orthonormal basis of Hermitian matrices leads to the formulation of the matrix equation \eqref{eq:inf0} as a linear system equation.
Details can be found in \cite{onoTBDmedian,onomasterthesis}.

In the next simulations, $m$ number of $N$-dimensional data are sampled from $CN \left(\bm{0}, \bm{\Sigma} \right)$ given by Section \ref{sec: DP}. 
The matrix $ \overline{\bm{R}}$ are computed from the $m$ autocovariance matrices $\{\bm{R}_i\}_{i=1}^m$ of the sampling data. 
Next, $n$ outliers $\{ \bm{P}_j\}_{j=1}^n$ are mixed with $\{\bm{R}_{i}\}_{i=1}^m$ and are modelled as the autocovariance matrix of the signal, $ a \bm{s} + \bm{c}$, where $\bm{s}$ is the target in \eqref{eq: target}, $\bm{c}$ is the Gaussian clutter of $CN \left(\bm{0}, \bm{\Sigma} \right)$, 
and $a$ is calculated from the SCR which is 40 dB.
The numerical simulations are conducted in the setting of $m$ chosen as 50 and $n$ ranging from 1 to 40.

\Figref{fig: Influencefunction} shows the influence functions computed by averaging 1000 times and displayed in the log scale. 
The lower the influence functions are, the more robust the detector is.
Therefore, from \Figref{fig: Influencefunction}, the AIRM mean and the BW mean are robust among all Riemannian geometric means and medians.
Considering the detection performance and computational complexity, the BW mean is the most effective estimator among all Riemannian geometric detectors studied here.

\begin{figure}[htbp]
	\centering
	\includegraphics[width = \linewidth]{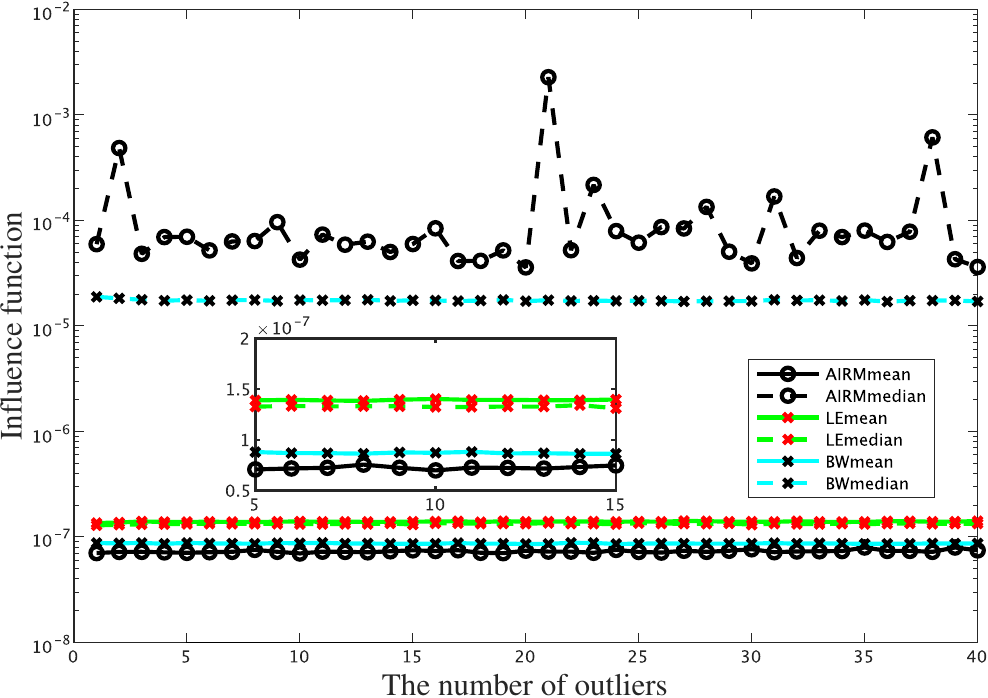}
	% \captionsetup{width=0.8\textwidth}
	\caption{Influence functions of the Riemannian geometric means and medians.}
	\label{fig: Influencefunction}
\end{figure}

%\begin{figure}[htbp]
%	\centering
%	\includegraphics[width = \linewidth]{processMCFAR.pdf}
%	\caption{The process of matrix-CFAR.}\label{fig: processMCFAR}
%\end{figure}

\section{CONCLUSION}
\label{sec: con}
The BW metric defined in HPD manifolds was applied to the matrix-CFAR detector and compared with other Riemannian metrics, including the AIRM and the LE metric, as well as the conventional AMF and ANMF. To derive the BW mean and median, Riemannian gradient descent algorithms were proposed and their computational complexity was investigated. 
It was shown that the Riemannian gradient descent algorithm was more convergent than the fixed-point algorithms which have been proposed elsewhere.
%To investigate the detection performance, numerical simulations were conducted and compared with the Riemannian geometric mean and median associated with the AIRM and the LE metric, as well as the conventional AMFs. 

Numerical simulations of the scenario using the ideal steering vector as the target showed that the matrix-CFAR with the BW metric outperformed all other detectors when the number of observation data is limited.
Additionally, it was observed that the detection performance of the AMF can be improved by replacing the SCM with the Riemannian geometric means and medians. 
In the signal mismatched scenario, it was noted that the matrix-CFAR is more robust than the AMF and ANMF.
Robustness analysis and simulations indicated that the BW mean is among the most robust ones to outliers, of a similar order as the AIRM mean. 

%Additionally, in order to analyse the robustness of the Riemannian geometric means and medians to outliers, the robustness analysis was performed by the numerical simulation of the influence functions.
%By simulations, the AIRM mean is the most robust among all of them, and the BW mean is the second best. 
As future research,  we plan to explore the practical applications of the matrix-CFAR technique on real-world field data. 
Secondly, Gaussian distributions for HPD matrices can be established through a symplectic model introduced by \cite{barbaresco2022symplectic, barbaresco2021gaussian}, which may be applied in determining the threshold of matrix-CFAR.
It would also be interesting to investigate other estimators for the autocovariance matrix, such as structured autocovariance interference \cite{COP2016} and the diagonal loading (e.g., \cite{8450037,4383590}).

\section*{APPENDIX}
% \renewcomand{\thesubsection}{\Alph{subsection}}\varepsilon}
% \label{subsec: proofofSyl}
In this appendix, we show the equivalence of the equations \eqref{eq: Ap1} and \eqref{eq: Ap2}.
The Sylvester equation reads
    \begin{equation*}
        \bm{A} \bm{X} + \bm{X} \bm{B} = \bm{C},
    \end{equation*}
    where $\bm{A} = (a_{ij})$ is an $m\times m$ matrix, $\bm{B}$ is an $n\times n$ matrix, $\bm{C}$ is an $m\times n$ matrix which are given and the $m\times n$ matrix $\bm{X}$ is the unknown.
    % Using the Kronecker product $\bm{A} \otimes \bm{B} = \left\{ a_{ij} \bm{B} \right\}_{i,j= 1}^n$ and the vectorization operator of a matrix, denoted by $\vect$, which stacks the columns of a matrix into one long vector,
    It can be rewritten as  
    \begin{equation*} 
        \left( \bm{I}_m \otimes \bm{A} + \bm{B}\trans \otimes \bm{I}_n \right) \vect \bm{X} = \vect \bm{C},
    \end{equation*}
    where $\bm{I}_m$ and $\bm{I}_n$ are the $m$- and $n$-dimensional identity matrices, the Kronecker product is $\bm{A} \otimes \bm{B} = \left\{ a_{ij} \bm{B} \right\}_{i,j= 1}^m$ and the vectorization operator of a matrix is denoted by $\vect$  which stacks the columns of a matrix into one long vector.
    Assuming that $ m = n $ and $\bm{C} = \bm{0}$, the Sylvester equation becomes
    \begin{equation} \label{eq : vecSylvester}
        \left( \bm{I}_n \otimes \bm{A} + \bm{B}\trans \otimes \bm{I}_n \right) \vect \bm{X} =0.
    \end{equation}
    Let $ \operatorname{eig} \left( \bm{A} \right) = \left\{ \lambda_i \right\}_{i = 1}^n $ and $ \operatorname{eig} \left( \bm{B} \right) = \left\{ \mu_i \right\}_{i = 1}^n $. 
    The eigenvalues of the coefficient matrix are 
    \begin{equation*}
        \operatorname{eig} \left( \bm{I}_n \otimes \bm{A} + \bm{B}\trans \otimes \bm{I}_n \right) = \left\{ \lambda_i + \mu_j \mid i,j = 1, \ldots, n \right\}.
    \end{equation*}
    Now we consider the equation \eqref{eq: Ap1}, namely
    \begin{equation} \label{eq : gradLyap}
        \sum\limits_{i=1}^{m} \bm{R} \left( \bm{I} - \bm{R}_i \# \bm{R}^{-1} \right) + \sum\limits_{i=1}^{m} \left( \bm{I} - \bm{R}_i \# \bm{R}^{-1} \right) \bm{R} = \bm{0}.
    \end{equation}
    Plugging $ \bm{A} = \bm{R}, \bm{B} = \bm{R} $ and $\bm{X} = \sum\limits_{i=1}^{m} \left( \bm{I} - \bm{R}_i \# \bm{R}^{-1} \right) $ into the equation \eqref{eq : vecSylvester}, we have
    \begin{equation*} \label{eq : vecgrad}
        \left( \bm{I}_n \otimes \bm{R} + \bm{R}\trans \otimes \bm{I}_n \right) \vect \left( \sum\limits_{i=1}^{m} \left( \bm{I} - \bm{R}_i \# \bm{R}^{-1} \right) \right) = \bm{0}.
    \end{equation*}
    Then the eigenvalues of the coefficient matrix are given by
    \begin{equation*}
        %\operatorname{eig} \left( \bm{I}_n \otimes \bm{R} -  \left(-\bm{R}\right)\hermconj \otimes \bm{I}_n \right) = \left\{ \lambda_i - (-\lambda_j) \mid i,j = 1, \ldots, n \right\}
        \operatorname{eig} \left( \bm{I}_n \otimes \bm{R} + \bm{R}\trans \otimes \bm{I}_n \right) = \left\{ \lambda'_i + \lambda'_j \mid i,j = 1, \ldots, n \right\} ,
    \end{equation*}
    where $\lambda'_i $ are the eigenvalues of the matrix $\bm{R}$. Since the matrix $\bm{R}$ is positive definite, the matrix $ \bm{I}_n \otimes \bm{R} + \bm{R}\trans \otimes \bm{I}_n $ is also positive definite.
    Consequently, the equation \eqref{eq : gradLyap} becomes
    \begin{equation*} 
        \sum\limits_{i=1}^{m} \left( \bm{I} - \bm{R}_i \# \bm{R}^{-1} \right) = \bm{0}.
    \end{equation*}

\bibliographystyle{IEEEtran}
\bibliography{mybibfile}

%\begin{IEEEbiography}[{\includegraphics[width=1in,height=1.25in, clip, keepaspectratio]{author/YOno}}]{Yusuke Ono} (Member, IEEE) received his Bachelor's degree and Master's degree from the Department of Mechanical Engineering, Keio University, Tokyo, Japan, in 2021 and 2023, respectively. He is currently pursuing his Ph.D. degree at the same institute. He is interested in data analysis, mathematical optimization on manifolds and information geometry.
%  \end{IEEEbiography}
%
%  \begin{IEEEbiography}[{\includegraphics[width=1in,height=1.25in, clip, keepaspectratio]{author/LPeng}}]{Linyu Peng} (Member, IEEE) earned his Ph.D. degree in mathematics from the University of Surrey, Guildford, U.K., in 2013. From 2013 to 2020, he worked at Waseda University, Tokyo, Japan. He is currently an associate professor at the Department of Mechanical Engineering and Center for Applied and Computational Mechanics, Keio University, Tokyo, Japan. His research interests encompass both pure and applied mathematics, with a focus on mathematical optimization, algebraic and geometric theories of differential equations and finite difference equations, as well as information geometry.
%  \end{IEEEbiography}

\end{document}